\documentclass[11pt,reqno]{amsart}

\usepackage[margin=25mm,a4paper]{geometry}
\usepackage{mathtools}
\usepackage{amsaddr}
\usepackage{bbm}
\usepackage[shortlabels]{enumitem}
\usepackage{cite}

\usepackage{hyperref}

\numberwithin{equation}{section}
\allowdisplaybreaks

\newcommand{\R}{\mathbb R}
\newcommand{\Prob}{\mathbf P}
\newcommand{\entropy}{\mathbf H}
\newcommand{\ex}{\mathbf E}
\newcommand{\1}{\mathbbm{1}}
\DeclareMathOperator{\Var}{Var}

\theoremstyle{definition}
\newtheorem{definition}{Definition}[section]
\theoremstyle{plain}
\newtheorem{theorem}[definition]{Theorem}
\newtheorem{proposition}[definition]{Proposition}
\newtheorem{lemma}[definition]{Lemma}

\theoremstyle{remark}
\newtheorem{remark}[definition]{Remark}
\newtheorem{example}[definition]{Example}

\begin{document}

\title[Properties of the  Shannon,  R\'{e}nyi and other entropies]{Properties of the  Shannon,  R\'{e}nyi and other entropies: dependence in parameters, robustness in distributions and extremes}

\author{Iryna Bodnarchuk$^1$, Yuliya Mishura$^1$, Kostiantyn Ralchenko$^{1,2}$}

\address{$^1$Taras Shevchenko National University of Kyiv\\
$^2$University of Vaasa}

\thanks{YM is supported by The Swedish Foundation for Strategic Research, grant UKR24-0004, and by the Japan Science and
Technology Agency CREST, project reference number JPMJCR2115.
KR is supported by the Research Council of Finland, decision number 359815.
YM and KR acknowledge that the present research is carried out within the frame and support of the ToppForsk project no.~274410 of the Research Council of Norway with the title STORM: Stochastics for Time-Space Risk Models.}

\begin{abstract}
We calculate and analyze various entropy measures and their properties for selected probability distributions. The entropies considered include Shannon, R\'enyi, generalized R\'enyi, Tsallis, Sharma--Mittal, and modified Shannon entropy, along with the Kullback--Leibler divergence. These measures are examined for several distributions, including gamma, chi-squared, exponential, Laplace, and log-normal distributions.
We investigate the dependence of the entropy on the parameters of the respective distribution. 
We also study the convergence of Shannon entropy for certain probability distributions.
Furthermore, we identify the extreme values of Shannon entropy for Gaussian vectors.
\end{abstract}

\maketitle

\section{Introduction}\label{secIntr}

The concept of entropy, being originated in thermodynamics, then became of high importance in biology, chemistry, information theory, synergetics and many other fields.  Unlike a large number of articles that examine all possible applications of the concept of entropy (among the most recent such works, we will name, without in any way claiming to be a complete list, such works as \cite{Bianconi2008,  CALI2017, Frenkel2015,HARRN2021,MilevInverardiTagliani2012, MRZZ_2024,Singh2013,Ullah1996,VanCampenhoutCover1981,YanLiSun2025,ZhangMengMaYinZhou2024}), in this paper we set ourselves the following task: before applying one or another entropy of one or another probability distribution, this selected entropy must first be calculated, represented clearly through the  parameters of the distribution, and its basic properties must be studied. The solution to this problem is what this article is dedicated to. Taking in mind that we have already made some related calculations, namely, in the paper \cite{MMRR} we calculated several entropies of the Gaussian distribution and studied their properties, in \cite{MMRS} we studied the Shannon entropy for the fractional Gaussian noise, and in the paper \cite{braiman2024} we investigated properties of the Shannon and  R\'enyi
entropies of the Poisson distribution, we concentrated our efforts on the calculations and properties that were not produced before. 

More precisely, in Section \ref{sec2} we consider 6 types of entropies, and all of them are named after those who introduced them into consideration: the Shannon entropy, the R\'{e}nyi entropy with index $\alpha>0,$ $\alpha\ne1,$ the generalized R\'enyi entropy with parameter $\alpha > 0$, the Tsallis entropy with parameter $\alpha > 0$, $\alpha \neq 1 $,  the generalized R\'enyi entropy with parameters $\alpha \neq \beta$, $\alpha > 0$, $\beta > 0$ and the Sharma--Mittal entropy with parameters $\alpha > 0$, $\beta > 0$, $\alpha \neq 1$, $\beta \neq 1 $. All of  them are defined both for discrete and absolutely continuous distributions. Additionally, we consider the Kullback--Leibler divergence for such distributions, and the modified Shannon  entropy for  absolutely continuous distributions with bounded density, and such modified entropy is always non-negative for them. 

Then, in Section \ref{expparam} we calculate the Shannon, all  R\'{e}nyi, Tsallis, Sharma--Mittal entropies and the Kullback--Leibler divergence for gamma, chi-squared, exponential, Laplace and log-normal distributions, and also calculate the modified Shannon entropy for normal,  gamma, chi-squared, exponential, Laplace and log-normal distributions, for such values of parameters of distributions that supply the boundedness of their density. 

Having these values, in Section \ref{secdeppar} we study the  dependence of entropies in parameters of the respective distribution. Without wishing to overload the reader with numerous technical details, but wishing to fill in the existing gaps, we will limit ourselves to the selected distributions and entropies. Namely, we provide the new proof of the fact that the Shannon entropy of the Poisson distribution increases with the intensity $\lambda$, study in detail the behavior of the Shannon entropy of gamma distribution as the function of two parameters, using properties of digamma and trigamma functions, and then analyze the behavior of all entropies for exponential  as the function of intensity $\lambda$. In principle, other distributions can be treated similarly. 

In Section \ref{robust} we prove that some convergences of distributions imply convergences of  Shannon entropies, and this sense, Shannon entropy has a property of robustness. Another approach to robustness of Shannon entropy was realized in \cite{MR4312787}.

Finally, in Section \ref{sec6} we consider Gaussian vector and its Shannon entropy. This entropy contains logarithm of determinant of covariance matrix. Assuming that we know   variances of vector components, in other words, we know the entropies of one-dimensional distributions, or, again in other words, the diagonal of covariance matrix,  we find  minimal and maximal values of the determinant and establish on which vectors they are realized.  

\section{Calculation of entropies, divergence and modified entropy}\label{sec2}

\subsection{Definitions of entropies} \looseness=1 Consider two kinds of probability distributions:  a discrete probability distribution with probabilities $\{p_i, i \geq 0\}$,
for which  we assume that they  are non-degenerate in the sense that all  probabilities  $p_i\in(0,1) $  and $\sum_{i = 0}^{\infty}p_i=1$, and probability distributions having a  density  $\{p(x), x\in \R\}$. Let us introduce six entropies and Kullback--Leibler divergence for such distributions. 

\enlargethispage{5pt}

\begin{definition}\label{def:entropies}
\begin{enumerate}
\item 
The Shannon entropy of a discrete distribution $\{p_i,i\geq0\}$ (of a distribution with the density $p(x),x\in\R$) is given by the formulas
\begin{gather}
\entropy_{SH}=\entropy_{SH} (p_i,i\geq0) =-\sum_{i = 0}^{\infty}  p_i\log {p_i}  =  \sum_{i = 0}^{\infty}   p_i \log\left(\frac{1}{p_i}\right) ,
\notag\\
\label{eq:SHdens}
\entropy_{SH}=\entropy_{SH} (p(x),x\in\R) 
=-\int_\R p(x) \log p(x)dx  =  \int_\R  p(x) \log\left(\frac{1}{p(x)}\right)dx ,
\end{gather}
respectively.

\item
The R\'{e}nyi entropy with index $\alpha>0,$ $\alpha\ne1,$ of a discrete distribution $\{p_i,i\geq0\}$ (of a distribution with the density $p(x),x\in\R$) is given by the formulas
\begin{gather*}
\entropy_{R}(\alpha)=\entropy_{R}(\alpha, p_i, i\geq0) = \frac{1}{1-\alpha} \log
\left(\sum_{i=0}^\infty p_i^{\alpha}\right).
\\
\entropy_{R}(\alpha) =\entropy_{R}(\alpha,p(x),x\in\R)=
\frac{1}{1-\alpha} \log\left(\int_\R p^\alpha(x)dx\right),
\end{gather*}
respectively.
 \item  The generalized R\'enyi entropy with parameter $\alpha > 0$  of a discrete distribution $\{p_i,i\geq0\}$ (of a distribution with the density $p(x),x\in\R$) is given by the formulas
    \begin{gather}
    \entropy_{GR}(\alpha) = \entropy_{GR}(\alpha, p_i, i\geq0)
    = - \frac{\sum_{i = 0}^{\infty} p_i^\alpha   \log{p_i}  }{\sum_{i = 0}^{\infty} p_i^\alpha}
    = \frac{\sum_{i = 0}^{\infty} p_i^\alpha   \log{\frac{1}{p_i}}  }{\sum_{i = 0}^{\infty} p_i^\alpha},
    \notag
    \\\label{eq:GR2dens}
    \entropy_{GR}(\alpha) = \entropy_{GR}(\alpha,p(x),x\in\R)
    = - \frac{\int_\R p^\alpha(x)  \log p(x) dx}{\int_\R p^\alpha(x) dx},
    \end{gather}
     respectively. 
\item
    The Tsallis entropy with parameter $\alpha > 0$, $\alpha \neq 1 $,  of a discrete distribution $\{p_i,i\geq0\}$ (of a distribution with the density $p(x),x\in\R$) is given by the formulas
    \begin{gather}
   \entropy_T(\alpha) = \entropy_T(\alpha, p_i, i\geq0)
   = \frac{1}{1-\alpha} \left( \sum_{i = 0}^{\infty} p_i^\alpha - 1 \right),
   \notag
   \\\label{eq:Tdens}
   \entropy_T(\alpha) = \entropy_T(\alpha,p(x),x\in\R)
   = \frac{1}{1-\alpha} \left(\int_\R p^\alpha(x) dx - 1 \right),
    \end{gather}
 respectively. 
\item
    The generalized R\'enyi entropy with parameters $\alpha \neq \beta$, $\alpha > 0$, $\beta > 0$  of a discrete distribution $\{p_i,i\geq0\}$ (of a distribution with the density $p(x),x\in\R$) is given by the formulas
    \begin{gather}
    \entropy_{GR}(\alpha, \beta) = \entropy_{GR}(\alpha, \beta, p_i, i\geq0)
    = \frac{1}{\beta - \alpha}  \log\left(\frac{\sum_{i = 0}^{\infty} p_i^\alpha }{\sum_{i = 0}^{\infty} p_i^\beta}\right),
    \notag
    \\
    \label{eq:GR1dens}
    \entropy_{GR}(\alpha, \beta) = \entropy_{GR}(\alpha, \beta,p(x),x\in\R)
    = \frac{1}{\beta - \alpha}  \log\left(\frac{\int_\R p^\alpha(x) dx}{\int_\R p^\beta(x) dx}\right),
    \end{gather}
     respectively. 
\item
    The Sharma--Mittal entropy with parameters $\alpha > 0$, $\beta > 0$, $\alpha \neq 1$, $\beta \neq 1 $,  of a discrete distribution $\{p_i,i\geq0\}$ (of a distribution with the density $p(x),x\in\R$) is given by the formulas
    \begin{gather}
    \entropy_{SM}(\alpha, \beta) = \entropy_{SM}(\alpha, \beta, p_i, i\geq0)
    = \frac{1}{1-\beta} \left( \left( \sum_{i = 0}^{\infty} p_i^\alpha \right)^{\frac{1- \beta}{1 - \alpha}} - 1 \right),
    \notag
    \\\label{eq:SMdens}
\entropy_{SM}(\alpha, \beta) = \entropy_{SM}(\alpha, \beta,p(x),x\in\R)
    = \frac{1}{1-\beta} \left( \left(\int_\R p^\alpha(x) dx\right)^{\frac{1- \beta}{1 - \alpha}} - 1 \right),
    \end{gather}
   respectively. 
 \item The Kullback--Leibler divergence (relative entropy, Shannon relative entropy) between  two discrete probability distributions $\mathbf{P}$ and $\mathbf{Q}$ with probabilities $p_i,q_i, i\ge 0$ (between two probability distributions $\mathbf{P}$ and $\mathbf{Q}$ with densities $p(x),q(x), x\in \R$)  is given by the formulas
 $$\entropy_{SH}(\mathbf{P}||\mathbf{Q})=\sum_{i = 0}^{\infty} p_i\log\left(\frac{p_i}{q_i}\right),\quad
  \entropy_{SH}(\mathbf{P}||\mathbf{Q})=\int_\R p(x)\log\left(\frac{p(x)}{q(x)}\right) dx, $$ 
 respectively. 
\end{enumerate}
\end{definition}
\begin{remark}\label{remkulb} Recall that in any case it follows from the Jensen inequality that $$\entropy_{SH}(\mathbf{P}||\mathbf{Q})=\mathbf{E}_\mathbf{Q}\left(\frac{d\mathbf{P}}{d\mathbf{Q}}\log\left(\frac{d\mathbf{P}}{d\mathbf{Q}}\right)\right)\ge\mathbf{E}_\mathbf{Q}\left(\frac{d\mathbf{P}}{d\mathbf{Q}}\right)\log\left(\mathbf{E}_\mathbf{Q}\left(\frac{d\mathbf{P}}{d\mathbf{Q}}\right)\right) =0,$$
with the equality if and only if $\mathbf{Q}=\mathbf{P}.$
\end{remark} 
\subsection{Modified Shannon entropy  for the distributions with bounded density}\label{modify}
If to compare the formulas for the Shannon entropy for discrete distribution,
$\entropy_{SH} (p_i,i\ge 0)=-\sum p_i\log p_i$ and for distribution with density,
$\entropy_{SH} (p(x),$ $x\in\R)=-\int p(x)\log p(x) dx$,
they have distinct behaviour in the sense that $\entropy_{SH} (p_i,i\ge 0)$ is always strictly positive while $\entropy_{SH} (p(x),x\in\R)$ can be negative, zero or positive.
To avoid this disadvantage, we can introduce the modified Shannon entropy which will be strictly positive.

\begin{definition}
    Let $p=\{p(x),\ x\in\R\}$ is a bounded density of a probability distribution, and $M=\sup\limits_{x\in \R}p(x)<\infty$.
    Put $\widetilde{p}(x)=\dfrac{p(x)}{M}$. Then the modified Shannon entropy is defined as
    \begin{gather*}
    \entropy_{SH, M} (p(x),x\in\R)=-\int \widetilde{p}(x)\log \widetilde{p}(x) dx.
    \end{gather*}
\end{definition}

The following result is straightforward.
\begin{proposition}\label{Prop1}
    The modified Shannon entropy is nonnegative and equals
    \begin{equation}\label{modif}
       \entropy_{SH, M} (p(x),x\in\R)=\frac{1}{M}\entropy_{SH} (p(x),x\in\R)+\frac{\log M}{M}.
    \end{equation}
\end{proposition}

\section{Calculation of entropies of gamma,   exponential, chi-squared, Laplace and log-normal   distributions}\label{expparam}
Let us calculate all introduce entropies for several distributions: gamma, exponential, chi-squared, Laplace and log-normal. Also, we provide some values of the modified Shannon entropy. 
\subsection{Gamma distribution: entropies and the Kullback--Leibler divergence}\label{gamma}
\begin{definition}
    A random variable has gamma distribution with parameters $\mu>0$ and $\lambda>0$ if its density has a form
    $$
    p(x)=\frac{\lambda^\mu}{\Gamma(\mu)}x^{\mu-1}e^{-\lambda x}\1_{(0,+\infty)}(x).
    $$
\end{definition}
\begin{proposition}\label{gamma-entropies}
\begin{enumerate}[{1)}]
    \item The Shannon entropy of the gamma distribution is given by the formula 
    \begin{equation}\label{eq:gamsh}
        \entropy_{SH}^{(gamma)}(\lambda,\mu)=-\log \lambda + \log \Gamma(\mu)+\mu
        -\psi(\mu)(\mu-1),
    \end{equation}
    where $\psi(\mu)=\frac{\Gamma'(\mu)}{\Gamma(\mu)}$ denotes the digamma function.
    
\item The R\'enyi entropy with parameter $\alpha>0$, $\alpha \ne 1$, of the gamma distribution is well-defined for  any $\alpha>0$ if $\mu\ge 1$ and for $0<\alpha<\frac{1}{1-\mu}$
if $0<\mu< 1$. In other words, it is defined for $\alpha(\mu-1)>-1$. 
For these parameter values, the R\'enyi entropy is given by the formula 
    \begin{gather}\label{eq:gamr}
        \entropy_{R}^{(gamma)} (\alpha,\lambda,\mu)
        =-\log \lambda -\frac{1}{1-\alpha}
        \log\frac{\Gamma^\alpha(\mu)}{\Gamma(\alpha(\mu-1)+1)}
        -\frac{1-\alpha+\alpha\mu}{1-\alpha}\log\alpha.
    \end{gather}

\item The generalized R\'enyi entropy with parameter $\alpha$ of the gamma distribution is well-defined under the same conditions on $\alpha$ and $\mu$ as for the R\'enyi entropy.
It is expressed as
    \begin{multline}
        \entropy_{GR}^{(gamma)} (\alpha,\lambda,\mu) =-\log\lambda +\log\Gamma(\mu)+(\mu-1)\log\alpha\\
        - (\mu-1)\psi(\alpha(\mu-1)+1)+\mu-1+\frac{1}{\alpha}.\label{eq:gamgr}
    \end{multline}

\item
The Tsallis entropy with parameter $\alpha>0$, $\alpha\neq 1$ of the gamma distribution is given by the formula 
    \begin{equation*}
    \entropy_{T}^{(gamma)}(\alpha,\lambda,\mu)
    =\frac{1}{1-\alpha}\left(\lambda^{\alpha-1}
    \frac{\alpha^{\alpha(1-\mu)-1}\Gamma(\alpha(\mu-1)+1)}{\Gamma^\alpha(\mu)}-1\right).
    \end{equation*}

\item
The generalized R\'{e}nyi entropy with parameters $\alpha,\beta>0, \alpha\neq \beta$ of the gamma distribution is well-defined for $\min{(\alpha(\mu-1),\beta(\mu-1))}>-1$ and is given by the formula 
\begin{gather*}
\entropy_{GR}^{(gamma)}(\alpha,\beta,\lambda,\mu)
=\frac{1}{\beta-\alpha}\log\left(\lambda^{\alpha-\beta}
\frac{\beta^{\beta(\mu-1)+1}}{\alpha^{\alpha(\mu-1)+1}}
\Gamma^{\beta-\alpha}(\mu)
\frac{\Gamma(\alpha(\mu-1)+1)}{\Gamma(\beta(\mu-1)+1)}\right)\nonumber\\
=-\log\lambda+\frac{1}{\beta-\alpha}\log\left(
\frac{\beta^{\beta(\mu-1)+1}}{\alpha^{\alpha(\mu-1)+1}}
\frac{\Gamma(\alpha(\mu-1)+1)}{\Gamma(\beta(\mu-1)+1)}\right)
+\log\Gamma(\mu).
\end{gather*}

\item
The Sharma--Mittal entropy with parameters $\alpha,\beta>0, \alpha\neq \beta, \alpha\neq 1, \beta\neq1$ of the gamma distribution is well-defined for
$\alpha(\mu-1)>-1$  and is given by the formula 
\begin{equation*}
\entropy_{SM}^{(gamma)}(\alpha,\beta,\lambda,\mu)
= \frac{1}{1-\beta}\left(
\frac{\lambda^{\beta-1}\alpha^{\frac{(\alpha(1-\mu)-1)(1-\beta)}{1-\alpha}}}{\Gamma^{\frac{\alpha(1-\beta)}{1-\alpha}}(\mu)}
\Gamma^{\frac{1-\beta}{1-\alpha}}(\alpha(\mu-1)+1)
-1\right).
\end{equation*}

\item
The  Kullback--Leibler divergence (relative entropy) of two gamma distributions is given by the formula 
\begin{equation*}
\entropy_{KL}^{(gamma)} (\lambda,\lambda_1,\mu,\mu_1)
=\mu_1\log\frac{\lambda}{\lambda_1}+\mu\left(\frac{\lambda_1}{\lambda}-1\right)+\log\frac{\Gamma(\mu_1)}{\Gamma(\mu)}+(\mu-\mu_1)\psi(\mu).
\end{equation*}
\end{enumerate}
\end{proposition}

\begin{remark}\label{gammakulb} 
According to Remark \ref{remkulb}, the expression
\begin{equation}
\mu_1\log\frac{\lambda}{\lambda_1} + \mu\left(\frac{\lambda_1}{\lambda} - 1\right) + \log\frac{\Gamma(\mu_1)}{\Gamma(\mu)} + (\mu - \mu_1)\psi(\mu) \geq 0
\label{kulkul}
\end{equation}
holds for any strictly positive parameters and equals zero if and only if 
$\lambda=\lambda_1$ and $\mu=\mu_1$. 
This result can be established analytically as follows.
First, let 
$x = \frac{\lambda}{\lambda_1}$, and consider the function
$f(x) = \mu_1\log x+\mu\left(\frac{1}{x}-1\right)$, $x>0$,
which achieves its unique minimum at 
$x = \frac{\mu}{\mu_1}$. This can be verified by straightforward differentiation. Substituting this result leads to the inequality
\begin{multline}\label{gammakulb-proof1}
\mu_1\log\frac{\lambda}{\lambda_1} + \mu\left(\frac{\lambda_1}{\lambda} - 1\right) + \log\frac{\Gamma(\mu_1)}{\Gamma(\mu)} + (\mu - \mu_1)\psi(\mu)
\\*
\geq \mu_1 \log\frac{\mu}{\mu_1} + \mu\left(\frac{\mu_1}{\mu} - 1\right) + \log\frac{\Gamma(\mu_1)}{\Gamma(\mu)} + (\mu - \mu_1)\psi(\mu) \eqqcolon g(\mu, \mu_1), 
\end{multline} 
which becomes an equality if and only if 
$\frac{\lambda}{\lambda_1} = \frac{\mu}{\mu_1}$.
To establish the non-negativity of the function $g(\mu, \mu_1)$, the right-hand side of \eqref{gammakulb-proof1}, we minimize it with respect to $\mu$.
The partial derivative is given by 
\begin{equation}\label{gammakulb-proof2} \frac{\partial g(\mu, \mu_1)}{\partial \mu} = \frac{\mu_1}{\mu} - 1 - \psi(\mu) + \psi(\mu) + (\mu - \mu_1)\psi'(\mu) = (\mu - \mu_1)\left(\psi'(\mu) - \frac{1}{\mu}\right).
\end{equation}
Next, note that $\psi'(\mu) - \frac{1}{\mu} > 0$, because, using the series representation for the trigamma function $\psi'$ (see \cite[formula 8.363.8]{GradshteynRyzhik2007}), we get
\begin{align*}
\psi'(\mu) &= \frac{1}{\mu^2} + \frac{1}{(\mu+1)^2} + \frac{1}{(\mu+2)^2} + \ldots \\
&> \int_{\mu}^{\mu+1}\frac{\mathrm{d}x}{x^2} + \int_{\mu+1}^{\mu+2}\frac{\mathrm{d}x}{x^2} + \int_{\mu+2}^{\mu+3}\frac{\mathrm{d}x}{x^2} + \ldots 
= \int_{\mu}^{\infty}\frac{\mathrm{d}x}{x^2} = \frac{1}{\mu}. 
\end{align*}
Therefore, from \eqref{gammakulb-proof2}, we conclude that 
\[ 
\min_{\mu > 0} g(\mu, \mu_1) = g(\mu_1, \mu_1) = 0,
\]
which completes the proof of \eqref{kulkul}.
\end{remark}

\begin{proof}[Proof of Proposition \ref{gamma-entropies}]
    1) Substituting the density of the gamma distribution into the definition of the Shannon entropy, we get 
    \begin{gather*}
        \entropy_{SH}^{(gamma)}(\lambda,\mu)
        =-\frac{\lambda^\mu}{\Gamma(\mu)}\int_{0}^{\infty}x^{\mu-1}e^{-\lambda x}\log\left(\frac{\lambda^\mu}{\Gamma(\mu)}x^{\mu-1}e^{-\lambda x}\right)dx
        =I_1+I_2+I_3,
    \end{gather*}
where
    \begin{align}
        I_1&=-\frac{\lambda^\mu}{\Gamma(\mu)}\int_{0}^{\infty}x^{\mu-1}e^{-\lambda x}\log\left(\frac{\lambda^\mu}{\Gamma(\mu)}\right)dx
        =-\log\left(\frac{\lambda^\mu}{\Gamma(\mu)}\right),\label{eq:gamshI1}
        \\
        I_2&=-\frac{\lambda^\mu (\mu-1)}{\Gamma(\mu)}\int_{0}^{\infty}x^{\mu-1}e^{-\lambda x}\log{x}\, dx\nonumber\\
        &= -\frac{\mu-1}{\Gamma(\mu)}
        \int_{0}^{\infty}y^{\mu-1}e^{-y}(\log{y}-\log{\lambda})\, dy\nonumber\\
        &=-\frac{\mu-1}{\Gamma(\mu)}
        \int_{0}^{\infty}y^{\mu-1}e^{-y}\log{y}\, dy +(\mu-1)\log{\lambda} \nonumber\\
        &=-(\mu-1)\frac{\Gamma'(\mu)}{\Gamma(\mu)}+(\mu-1)\log{\lambda}
        =-(\mu-1)\psi(\mu)+(\mu-1)\log{\lambda},\label{eq:gamshI2}
    \end{align}
and
    \begin{equation}\label{eq:gamshI3}
        I_3=-\frac{\lambda^\mu}{\Gamma(\mu)}\int_{0}^{\infty}x^{\mu-1}e^{-\lambda x}(-\lambda x)dx
        =\frac{1}{\Gamma(\mu)}\Gamma(\mu+1)=\mu.
    \end{equation}
Now,  \eqref{eq:gamsh} immediately follows from \eqref{eq:gamshI1}--\eqref{eq:gamshI3}.

    2) Let $\mu\ge 1,\ \alpha>0$ or
    $0<\mu <1,\ 0<\alpha<\frac{1}{1-\mu}$. First, let us calculate \begin{align*}
        J &:= \int_0^{\infty}p^\alpha(x)dx 
        =\frac{\lambda^{\mu\alpha}}{\Gamma^\alpha(\mu)}\int_0^{\infty} x^{\alpha(\mu-1)} e^{-\lambda\alpha x}dx
        \notag\\
        &=\frac{\lambda^{\mu\alpha}}{\Gamma^\alpha(\mu)}
        (\lambda\alpha)^{-\alpha(\mu-1)-1}
        \int_0^{\infty} y^{\alpha(\mu-1)} e^{-y}dy
        \notag\\
        &=\frac{\lambda^{\mu\alpha}}{\Gamma^\alpha(\mu)}
        (\lambda\alpha)^{-\alpha(\mu-1)-1}
        \Gamma(\alpha(\mu-1)+1)
        =\lambda^{\alpha-1}\frac{\alpha^{-\alpha(\mu-1)-1}\Gamma(\alpha(\mu-1)+1)} {\Gamma^\alpha(\mu)}.
        \label{mainint}
    \end{align*}
Then
    \begin{gather*}
        \entropy_{R}^{(gamma)}(\alpha,\lambda,\mu)
        =\frac{1}{1-\alpha}\log\left(\int_0^{\infty}p^\alpha(x)dx\right)
        \\=-\log\lambda-\frac{\alpha}{1-\alpha}\log\Gamma(\mu)
        +\frac{1}{1-\alpha}\log\left(\Gamma(\alpha(\mu-1)+1)\right)
        -\frac{\alpha(\mu-1)+1}{1-\alpha}\log\alpha,
    \end{gather*}
and we get  \eqref{eq:gamr}.

    3) Let us calculate 
    $J_1=\int_0^{\infty}p^\alpha(x)\log p(x) dx$.  
Similarly to the proof of the statement 1), we can decompose $J_1$ as follows: $J_1=L_1+L_2+L_3$, where
    \begin{align*}
        L_1&=\frac{\lambda^{\mu\alpha}}{\Gamma^\alpha(\mu)}
        \log\frac{\lambda^{\mu}}{\Gamma(\mu)}
        \int_0^{\infty}x^{\alpha(\mu-1)} e^{-\lambda\alpha x}dx
        = J \log\frac{\lambda^{\mu}}{\Gamma(\mu)},\\
        L_2&=\frac{(\mu-1)\lambda^{\mu\alpha}}{\Gamma^\alpha(\mu)}
        \int_0^{\infty}x^{\alpha(\mu-1)} e^{-\lambda\alpha x}\log x dx\\
        &=\frac{(\mu-1)\lambda^{\mu\alpha}}{\Gamma^\alpha(\mu)}
        (\lambda\alpha)^{-\alpha(\mu-1)-1}
        \int_0^{\infty}y^{\alpha(\mu-1)} e^{-y}\log ydy\\
        &\quad-\frac{(\mu-1)\lambda^{\mu\alpha}}{\Gamma^\alpha(\mu)}\log(\lambda\alpha)
        (\lambda\alpha)^{-\alpha(\mu-1)-1}
        \int_0^{\infty}y^{\alpha(\mu-1)} e^{-y} dy\\
        &=\frac{(\mu-1)\lambda^{\mu\alpha}}{\Gamma^\alpha(\mu)}
        (\lambda\alpha)^{-\alpha(\mu-1)-1}
        \Gamma'(\alpha(\mu-1)+1)\\
        &\quad-\frac{(\mu-1)\lambda^{\mu\alpha}}{\Gamma^\alpha(\mu)}\log(\lambda\alpha)
        (\lambda\alpha)^{-\alpha(\mu-1)-1}
        \Gamma(\alpha(\mu-1)+1),
    \end{align*}
and
    \begin{gather*}
        L_3=-\frac{\lambda^{\mu\alpha+1}}{\Gamma^\alpha(\mu)}
        \int_0^{\infty}x^{\alpha(\mu-1)+1} e^{-\lambda\alpha x} dx
        =-\frac{\lambda^{\mu\alpha+1}}{\Gamma^\alpha(\mu)}
        (\lambda\alpha)^{-\alpha(\mu-1)-2}
        \Gamma(\alpha(\mu-1)+2).
    \end{gather*}
Then
    \begin{align*}
        \MoveEqLeft\entropy_{GR}^{(gamma)} (\alpha,\lambda,\mu) = -\frac{J_1}{J}\\*
        &=-\log\frac{\lambda^{\mu}}{\Gamma(\mu)}
        -(\mu-1)\frac{\Gamma'(\alpha(\mu-1)+1)}{\Gamma(\alpha(\mu-1)+1)}
        +(\mu-1)\log(\lambda\alpha)
        +\frac{\alpha(\mu-1)+1}{\alpha}\\
        &=-\log\lambda+\log\Gamma(\mu)+(\mu-1)\log\alpha
        -(\mu-1)\psi(\alpha(\mu-1)+1)+\mu-1+\frac{1}{\alpha},
    \end{align*}
which completes the proof of \eqref{eq:gamgr}.

    The statements 4)--6) immediately follow from \eqref{mainint}.

    7) We can proceed with the following relations
    \begin{align*}
    \entropy_{KL}^{(gamma)} (\lambda,\lambda_1,\mu,\mu_1)
    &=\frac{\lambda^\mu}{\Gamma(\mu)}
    \int_0^\infty x^{\mu-1}e^{-\lambda x}
    \log\left(\frac{\lambda^\mu}{\lambda_1^{\mu_1}}
    \frac{\Gamma(\mu_1)}{\Gamma(\mu)} x^{\mu-\mu_1} e^{-(\lambda-\lambda_1)x}
    \right)dx\\
    &=M_1+M_2+M_3,
    \end{align*}
where
\begin{align*}
    M_1&=\log\frac{\lambda^\mu \Gamma(\mu_1)}{\lambda_1^{\mu_1} \Gamma(\mu)},\\
    M_2&=\frac{(\mu-\mu_1)\lambda^\mu}{\Gamma(\mu)}
    \int_0^\infty x^{\mu-1}e^{-\lambda x}\log x dx
    =(\mu-\mu_1)\frac{\Gamma'(\mu)}{\Gamma(\mu)}
    -(\mu-\mu_1)\log\lambda\\
    &=(\mu-\mu_1)\psi(\mu)-(\mu-\mu_1)\log\lambda,\\
    M_3&=-(\lambda-\lambda_1)\frac{\lambda^\mu}{\Gamma(\mu)}
    \int_0^\infty x^{\mu}e^{-\lambda x}dx
    =-\left(1-\frac{\lambda_1}{\lambda}\right)\frac{\Gamma(\mu+1)}{\Gamma(\mu)}
    =-\mu\left(1-\frac{\lambda_1}{\lambda}\right),
\end{align*}
whence the proof follows.
\end{proof}

\subsection{Exponential distribution: entropies and the Kullback--Leibler divergence}\label{exp}
\begin{definition}
    A random variable $X_{\lambda}^{(exp)}$ has exponential distribution with parameter $\lambda>0$,
    if its density has a form 
    \begin{gather*}
    p_{\lambda}(x)=\lambda e^{-\lambda x}\1_{(0,+\infty)}(x).
    \end{gather*}    
\end{definition}

Recall that the exponential distribution is a special case of the gamma distribution with $\mu = 1$. Consequently, as a corollary of Proposition~\ref{gamma-entropies}, we obtain the following result.
\begin{proposition}
\begin{enumerate}[1)]
\item
The Shannon entropy of the exponential distribution equals
\begin{equation}
\entropy_{SH}^{(exp)} (\lambda)
=\log{\frac{e}{\lambda}}=-\log\lambda+1.\label{eq:eshexp}
\end{equation}

\item
The R\'{e}nyi entropy with parameter $\alpha>0$ of the exponential distribution equals
\begin{equation*}
\entropy_{R}^{(exp)}(\alpha,\lambda)
=\frac{1}{1-\alpha} \log\left(\frac{\lambda^{\alpha-1}}{\alpha}\right)=-\log\lambda+\frac{\log\alpha}{\alpha-1}.
\end{equation*}

\item
The generalized R\'{e}nyi entropy with parameter $\alpha>0$ of the exponential distribution equals
\begin{equation*}
\entropy_{GR}^{(exp)}(\alpha,\lambda)
=-\log\lambda+\frac{1}{\alpha}.
\end{equation*}

\item
The Tsallis entropy with parameter $\alpha>0$, $\alpha\neq 1$ of the exponential distribution equals
\begin{equation}
\entropy_{T}^{(exp)}(\alpha,\lambda)
=\frac{\lambda^{\alpha-1}-\alpha}{\alpha(1-\alpha)}.\label{eq:etexp}
\end{equation}

\item
The generalized R\'{e}nyi entropy with parameters $\alpha,\beta>0$, $\alpha\neq \beta$ of the exponential distribution equals
\begin{equation*}
\entropy_{GR}^{(exp)}(\alpha,\beta,\lambda)
=-\log\lambda+\frac{1}{\beta-\alpha}\log\frac{\beta}{\alpha}.
\end{equation*}

\item
The Sharma--Mittal entropy with parameters $\alpha,\beta>0, \alpha\neq \beta, \alpha\neq 1, \beta\neq1$ of the exponential distribution equals
\begin{align}
\entropy_{SM}^{(exp)}(\alpha,\beta,\lambda)
=\frac{1}{1-\beta}\left(\lambda^{\beta-1}\alpha^{\frac{1-\beta}{\alpha-1}}-1\right). \label{eq:esmexp2}
\end{align}

\item
The  Kullback--Leibler divergence (relative entropy)  of two exponential distributions equals
\begin{equation*}
\entropy_{KL}^{(exp)} (\lambda,\lambda_1)
=\log\left(\frac{\lambda}{\lambda_1}\right) + \frac{\lambda_1}{\lambda} - 1,
\end{equation*}
and $\entropy_{KL}^{(exp)} (\lambda_1,\lambda_1)=0$.
\end{enumerate}
\end{proposition}

\subsection{Chi-squared distribution: entropies and the Kullback--Leibler divergence}\label{chisq}

The chi-squared distribution with $\nu$ degrees of freedom represents the distribution of the sum of the squares of $\nu$ independent standard normal random variables. It is well known that this distribution is a special case of the gamma distribution, where the parameters are $\lambda = \frac12$ and $\mu = \frac{\nu}{2}$. Consequently, the following result arises as a direct application of Proposition~\ref{gamma-entropies}.

\begin{proposition}
\begin{enumerate}[{1)}]
    \item The Shannon entropy of the chi-squared distribution is given by
    \begin{equation*}
        \entropy_{SH}^{(\chi^2)} (\nu) = \log 2 + \log \Gamma(\tfrac{\nu}{2})+\tfrac{\nu}{2}
        -\psi(\tfrac{\nu}{2})(\tfrac{\nu}{2}-1).
    \end{equation*}
    
\item The R\'enyi entropy with parameter $\alpha>0$, $\alpha \ne 1$, of the chi-squared distribution is well-defined for  any $\alpha>0$ if $\nu\ge 2$ and for $0<\alpha<2$
if $\nu = 1$, or, in other words, for $\alpha(\frac{\nu}{2}-1)>-1$. 
For these parameter values, the R\'enyi entropy is given by
    \begin{gather*}
        \entropy_{R}^{(\chi^2)} (\alpha,\nu)
        =\log 2 -\frac{1}{1-\alpha}
        \log\frac{\Gamma^\alpha(\frac{\nu}{2})}{\Gamma(\alpha(\frac{\nu}{2}-1)+1)}
        -\frac{1-\alpha+\frac{\alpha\nu}{2}}{1-\alpha}\log\alpha.
    \end{gather*}

\item The generalized R\'enyi entropy with parameter $\alpha$ of the chi-squared distribution is well-defined under the same conditions on $\alpha$ and $\nu$ as for the R\'enyi entropy.
It is expressed as
    \[
        \entropy_{GR}^{(\chi^2)} (\alpha,\nu) =\log2 +\log\Gamma(\tfrac{\nu}{2})+(\tfrac{\nu}{2}-1)\log\alpha
        - (\tfrac{\nu}{2}-1)\psi(\alpha(\tfrac{\nu}{2}-1)+1)+\tfrac{\nu}{2}-1+\tfrac{1}{\alpha}.
    \]

\item
The Tsallis entropy with parameter $\alpha>0$, $\alpha\neq 1$ of the chi-squared distribution is given by
    \begin{equation*}
    \entropy_{T}^{(\chi^2)}(\alpha,\nu)
    =\frac{1}{1-\alpha}\left(2^{1-\alpha}\,
    \frac{\alpha^{\alpha(1-\frac{\nu}{2})-1}\Gamma(\alpha(\frac{\nu}{2}-1)+1)}{\Gamma^\alpha(\frac{\nu}{2})}-1\right).
    \end{equation*}

\item
The generalized R\'{e}nyi entropy with parameters $\alpha,\beta>0, \alpha\neq \beta$ of the chi-squared distribution is well-defined for $\min{(\alpha(\frac{\nu}{2}-1),\beta(\frac{\nu}{2}-1))}>-1$ and is given by 
\begin{gather*}
\entropy_{GR}^{(\chi^2)}(\alpha,\beta,\nu)
=\log2+\frac{1}{\beta-\alpha}\log\left(
\frac{\beta^{\beta(\frac{\nu}{2}-1)+1}}{\alpha^{\alpha(\frac{\nu}{2}-1)+1}}
\frac{\Gamma(\alpha(\frac{\nu}{2}-1)+1)}{\Gamma(\beta(\frac{\nu}{2}-1)+1)}\right)
+\log\Gamma\left (\frac{\nu}{2}\right ).
\end{gather*}

\item
The Sharma--Mittal entropy with parameters $\alpha,\beta>0, \alpha\neq \beta, \alpha\neq 1, \beta\neq1$ of the chi-squared distribution is well-defined for
$\alpha(\frac{\nu}{2}-1)>-1$  and is given by
\begin{equation*}
\entropy_{SM}^{(\chi^2)}(\alpha,\beta,\nu)
= \frac{1}{1-\beta}\left(
\frac{2^{1-\beta}\alpha^{\frac{(\alpha(1-\frac{\nu}{2})-1)(1-\beta)}{1-\alpha}}}{\Gamma^{\frac{\alpha(1-\beta)}{1-\alpha}}(\frac{\nu}{2})}
\Gamma^{\frac{1-\beta}{1-\alpha}}(\alpha(\tfrac{\nu}{2}-1)+1)
-1\right).
\end{equation*}

\item
The  Kullback--Leibler divergence (relative entropy) of two chi-squared distributions is given by
\begin{equation*}
\entropy_{KL}^{(\chi^2)} (\nu,\nu_1)
=\log\frac{\Gamma(\frac{\nu_1}{2})}{\Gamma(\frac{\nu}{2})}+\frac{\nu-\nu_1}{2}\psi\left(\frac{\nu}{2}\right).
\end{equation*}
\end{enumerate}
    
\end{proposition}
 \begin{remark} The non-negativity of the Kullback--Leibler divergence follows the inequalities:
\begin{align*}
\psi(\tfrac{\nu}{2}) & \ge \frac{\log \Gamma(\frac{\nu}{2})-\log \Gamma(\frac{\nu_1}{2})}{\frac{\nu}{2}-\frac{\nu_1}{2}}
\quad \text{for } \nu > \nu_1
\shortintertext{and}
\psi(\tfrac{\nu}{2}) & \le \frac{\log \Gamma(\frac{\nu_1}{2})-\log \Gamma(\frac{\nu}{2})}{\frac{\nu_1}{2}-\frac{\nu}{2}}
\quad \text{for } \nu < \nu_1,
\end{align*}
both of which are valid since the digamma function $\psi(x) = (\log \Gamma(x))'$ is strictly increasing. These inequalities follow directly from the application of the mean value theorem.    
 \end{remark}

\subsection{Laplace distribution: entropies and the Kullback--Leibler divergence}\label{Lapl}
\begin{definition}
A random variable has Laplace distribution with parameters $\mu\in\R$ and $\lambda>0$,
if its density has a form 
    \begin{gather*}
    p_{\mu,\lambda}(x)=\frac{\lambda}{2} e^{-\lambda |x-\mu|},\,x\in\R.
    \end{gather*}    
\end{definition}

\begin{proposition}
\begin{enumerate}[1)]
\item
The Shannon entropy of the Laplace distribution equals
\begin{equation*}
\entropy_{SH}^{(La)} (\lambda)
=-\log\frac{\lambda}{2}+1.
\end{equation*}

\item
The R\'{e}nyi entropy with parameter $\alpha>0$ of the Laplace distribution equals
\begin{equation*}
\entropy_{R}^{(La)}(\alpha,\lambda)
=-\log\frac{\lambda}{2}+\frac{\log\alpha}{\alpha-1}.
\end{equation*}

\item
The generalized R\'{e}nyi entropy with parameter $\alpha>0$ of the Laplace distribution equals
\begin{equation*}
\entropy_{GR}^{(La)}(\alpha,\lambda)
=-\log\frac{\lambda}{2}+\frac{1}{\alpha}.
\end{equation*}

\item
The Tsallis entropy with parameter $\alpha>0, \alpha\neq 1$ of the Laplace distribution equals
\begin{equation*}
\entropy_{T}^{(La)}(\alpha,\lambda)
=\frac{1}{\alpha(1-\alpha)}\left(\left(\frac{\lambda}{2}\right)^{\alpha-1}-\alpha\right).
\end{equation*}

\item
The generalized R\'{e}nyi entropy with parameters $\alpha,\beta>0, \alpha\neq \beta$ of the Laplace distribution equals
\begin{equation*}
\entropy_{GR}^{(La)}(\alpha,\beta,\lambda)
=-\log\frac{\lambda}{2}+\frac{1}{\beta-\alpha}\log\frac{\beta}{\alpha}.
\end{equation*}

\item
The Sharma--Mittal entropy with parameters $\alpha,\beta>0, \alpha\neq \beta, \alpha\neq 1, \beta\neq1$ of the Laplace distribution equals
\begin{equation*}
\entropy_{SM}^{(La)}(\alpha,\beta,\lambda)
=\frac{1}{1-\beta}\left(\lambda^{\beta-1}\alpha^{\frac{1-\beta}{\alpha-1}}-1\right).\label{eq:esmexp}
\end{equation*}

\item
The  Kullback--Leibler divergence (relative entropy) of two Laplace distributions equals
\begin{equation*}
\entropy_{KL}^{(La)} (\lambda,\lambda_1,\mu,\mu_1)=\log\frac{\lambda}{\lambda_1}
+ \frac{\lambda_1}{\lambda}\left(\lambda|\mu-\mu_1|+e^{-\lambda|\mu-\mu_1|}
\right) - 1,
\end{equation*}
and $\entropy_{KL}^{(La)} (\lambda,\lambda,\mu,\mu)=0$.\end{enumerate}
\end{proposition}
\begin{remark} Here the situation with the  Kullback--Leibler divergence is the following: the whole inequality 
 $$\log\frac{\lambda}{\lambda_1}
+ \frac{\lambda_1}{\lambda}\left(\lambda|\mu-\mu_1|+e^{-\lambda|\mu-\mu_1|}
\right) - 1 \ge 0$$ more likely was not considered before   for any strictly positive $\lambda,  \lambda_1, \mu, \mu_1$, while, it is very easy to get it if we just note that $$e^{-\lambda|\mu-\mu_1|}\ge 1- \lambda|\mu-\mu_1|,$$    and $$\log\frac{\lambda}{\lambda_1}+\frac{\lambda_1}{\lambda}-1\ge 0.$$  
\end{remark}
\begin{proof}
1) Indeed, the Shannon entropy $\entropy_{SH}^{((La))} (\lambda)$ can be evaluated as follows
\begin{gather*}
   \entropy_{SH}^{(La)} (\lambda)
   =-\int_{-\infty}^{+\infty}\frac{\lambda}{2}e^{-\lambda|x-\mu|}
   \log\left(\frac{\lambda}{2}e^{-\lambda|x-\mu|} \right)dx\\
   =\log \frac{\lambda}{2}
   +\int_{-\infty}^{+\infty}\frac{\lambda}{2}e^{-\lambda|x-\mu|}
   \lambda|x-\mu|dx
   =-\log\frac{\lambda}{2}+1.
\end{gather*}
Here we have used the fact that if $X$ has the Laplace distribution with parameters $\mu$ and $\lambda$, then the variable $Y=\lambda|X-\mu|$ has the exponential distribution with parameter~1, and the last integral is the mathematical expectation of $Y$.

2) In the following we will use the value of the integral
$J^{(La)}=\int_{-\infty}^{+\infty}p^\alpha_{\mu,\lambda}(x) dx$, so let we calculate it. We have that
\begin{gather}
    J^{(La)}=\frac{\lambda^{\alpha-1}}{2^{\alpha-1} \alpha}
    \int_{-\infty}^{+\infty}\frac{\lambda\alpha}{2}e^{-\lambda\alpha|x-\mu|}  dx
    =\frac{\lambda^{\alpha-1}}{2^{\alpha-1} \alpha}\label{eq:intJLa}.
\end{gather}

Thus, by \eqref{eq:intJLa} the R\'{e}nyi entropy with parameters $\alpha,\lambda$ of the Laplace distribution equals
\begin{gather*}
\entropy_{R}^{(La)}(\alpha,\lambda)
=\frac{1}{1-\alpha}\log\left(J^{(La)}\right)
=\frac{1}{1-\alpha} \log\left(\frac{\lambda^{\alpha-1}}{2^{\alpha-1}\alpha}\right)=-\log\frac{\lambda}{2}+\frac{\log\alpha}{\alpha-1}.
\end{gather*}

3) Let us calculate the integral
    $J_1^{(La)}=\int_0^{\infty}p^\alpha_{\mu,\lambda}(x)\log p_{\mu,\lambda}(x) dx$. Obtain,
    \begin{align*}
        J_1^{(La)}&=J^{(La)}\log\frac{\lambda}{2}-
        \frac{\lambda^{\alpha-1}}{2^{\alpha-1} \alpha^2}
    \int_{-\infty}^{+\infty}\frac{\lambda\alpha}{2}e^{-\lambda\alpha|x-\mu|}\lambda\alpha|x-\mu|  dx \\
        &=\frac{\lambda^{\alpha-1}}{2^{\alpha-1} \alpha}
        \left(\log\frac{\lambda}{2}-\frac{1}{\alpha}\right),
    \end{align*}
and the proof follows by \eqref{eq:GR2dens} and \eqref{eq:intJLa}.

The statements 4)--6) immediately follow from \eqref{eq:intJLa}.

7) We can proceed with the following relations
    \begin{align*}
    \entropy_{KL}^{(La)} (\lambda,\lambda_1,\mu,\mu_1)
    &=\int_{-\infty}^{+\infty}\frac{\lambda}{2}e^{-\lambda|x-\mu|}
   \log\left(\frac{\lambda}{\lambda_1}e^{-\lambda|x-\mu|+\lambda_1|x-\mu_1|} \right)dx    \\
    &=\log\frac{\lambda}{\lambda_1}-1
    +\int_{-\infty}^{+\infty}\frac{\lambda}{2}e^{-\lambda|x-\mu|}
   \lambda_1|x-\mu_1| dx.    
    \end{align*}

Consider the case $\mu\le\mu_1$. We can decompose the last integral as follows:
\begin{align*}
    J_2^{(La)}=\int_{-\infty}^{+\infty}\frac{\lambda}{2}e^{-\lambda|x-\mu|}
   \lambda_1|x-\mu_1| dx= - M_1^{(La)}-M_2^{(La)}+ M_3^{(La)},
\end{align*}
where
\begin{align*}
    M_1^{(La)}&=\int_{-\infty}^{\mu}\frac{\lambda}{2}e^{\lambda(x-\mu)}
   \lambda_1(x-\mu_1) dx
   =\frac{\lambda_1}{2\lambda}e^{-\lambda(\mu-\mu_1)}\int_{-\infty}^{\lambda(\mu-\mu_1)} ye^{y}dy\\
   &=\frac{\lambda_1}{2\lambda}\left(\lambda(\mu-\mu_1)-1\right),
\end{align*}
\begin{align*}
   M_2^{(La)}&=\int_{\mu}^{\mu_1}\frac{\lambda}{2}e^{-\lambda(x-\mu)}
   \lambda_1(x-\mu_1) dx
   =\frac{\lambda_1}{2\lambda}e^{-\lambda(\mu_1-\mu)}
   \int_{\lambda(\mu-\mu_1)}^{0} ye^{-y}dy\\
   &=\frac{\lambda_1}{2\lambda}\left(-e^{-\lambda(\mu_1-\mu)}+\lambda(\mu-\mu_1)+1\right),
   \\
   M_3^{(La)}&=\int_{\mu_1}^{+\infty}\frac{\lambda}{2}e^{-\lambda(x-\mu)}
   \lambda_1(x-\mu_1) dx
   =\frac{\lambda_1}{2\lambda}e^{-\lambda(\mu_1-\mu)}
   \int_{0}^{+\infty} ye^{-y}dy\\
   &=\frac{\lambda_1}{2\lambda} e^{-\lambda(\mu_1-\mu)},
\end{align*}
and we get
\begin{equation*}
    J_2^{(La)}=\frac{\lambda_1}{\lambda}\left(\lambda(\mu_1-\mu)+e^{-\lambda(\mu_1-\mu)} \right).
\end{equation*}

In the case $\mu_1<\mu$ we come to the similar  result, i.e.,
\begin{equation*}
    J_2^{(La)}=\frac{\lambda_1}{\lambda}\left(\lambda(\mu-\mu_1)+e^{-\lambda(\mu-\mu_1)} \right),
\end{equation*}
whence the proof follows.
\end{proof}

\subsection{Log-normal distribution: entropies and the Kullback--Leibler divergence}

\begin{definition}
A random variable $X_{m,\sigma^2}$ has lognormal distribution with parameters $m\in\R$ and $\sigma^2>0$, if $\log X_{m,\sigma^2}$ has a normal distribution $\mathcal{N}(m,\sigma^2)$.
\end{definition}

It is well known that the probability density function of $X_{m,\sigma^2}$ is given by
    \[
    p_{m,\sigma^2}(x) = \frac{1}{x \sigma \sqrt{2\pi}} \exp\left\{- \frac{(\log x - m)^2}{2\sigma^2}\right\} \1_{(0,+\infty)}(x).  
    \]     

To compute the associated entropies, we begin with the following auxiliary lemma, which provides the expectations of certain functions of the lognormal random variable $X_{m,\sigma^2}$.
    
\begin{lemma}
For all $p \in \R$
\begin{gather}
\label{exln1}
\ex X_{m,\sigma^2}^p = \exp\left\{mp + \frac{\sigma^2 p^2 }{2}\right\},
\\
\label{exln2}
\ex \left[ X_{m,\sigma^2}^p \log X_{m,\sigma^2}\right]
= \left(\sigma^2p + m\right)\exp\left\{mp + \frac{\sigma^2 p^2 }{2}\right\},
\\
\label{exln3}
\ex \left[ X_{m,\sigma^2}^p (\log X_{m,\sigma^2} - m)^2\right]
= \sigma^2 \left(\sigma^2 p^2 +1\right)\exp\left\{mp + \frac{\sigma^2 p^2 }{2}\right\},
\end{gather}
\end{lemma} 

\begin{proof}
Let $Z = \log X_{m,\sigma^2} - m \simeq \mathcal N(0,\sigma^2)$.
The moment generating function of $Z$ is given by
\[
\varphi(t) = \ex \exp\{tZ\} = \exp\left\{\frac{\sigma^2 t^2}{2}\right\},
\quad t\in\R.
\]
The expectations \eqref{exln1}--\eqref{exln3} can be expressed in terms of $\varphi$ and its derivatives:
\begin{gather*}
\ex X_{m,\sigma^2}^{p} = \ex e^{p\log X_{m,\sigma^2}} 
= \ex e^{p(Z+m)} = e^{mp}\varphi(p),
\\
\ex \left[ X_{m,\sigma^2}^p \log X_{m,\sigma^2}\right]
= \ex \left[e^{p(Z+m)} (Z + m)\right]
= e^{mp} \bigl(\varphi'(p) + m\varphi(p)\bigr),
\\
\ex \left[ X_{m,\sigma^2}^p (\log X_{m,\sigma^2} - m)^2\right]
= \ex \left[e^{p(Z+m)} Z^2\right]
= e^{mp} \varphi''(p).
\end{gather*}
Substituting 
$\varphi(p) = \exp \{\frac{\sigma^2 p^2}{2}\}$,
$\varphi'(p) = \sigma^2 p \exp \{\frac{\sigma^2 p^2}{2}\} $, and
$\varphi''(p)= \sigma^2(1 + \sigma^2 p^2)\exp \{\frac{\sigma^2 p^2}{2}\} $ 
into the above expressions yields \eqref{exln1}--\eqref{exln3}.
\end{proof}

\begin{proposition}\label{ln-entropies}
\begin{enumerate}[{1)}]
    \item The Shannon entropy of the lognormal distribution is given by the formula 
    \[
        \entropy_{SH}^{LN}\left(m,\sigma^2\right) =  \log\sigma + \frac12\log(2\pi) + m + \frac12.
    \]
    
\item The R\'enyi entropy with parameter $\alpha>0$, $\alpha \ne 1$, of the lognormal distribution is given by
\[
\entropy_{R}^{LN} (\alpha,m,\sigma^2) 
= \log\sigma + \frac12\log(2\pi) + m + \frac{\log\alpha}{2(\alpha-1)}
+ \frac{\sigma^2 (1-\alpha)}{2\alpha}.
\]

\item The generalized R\'enyi entropy with parameter $\alpha$ of the lognormal distribution is given by
\[
        \entropy_{GR}^{LN} (\alpha,m,\sigma^2)= \log\sigma + \frac12\log(2\pi) + m + \frac{1}{2\alpha} + \frac{\sigma^2\left(1-\alpha^2\right)}{2\alpha^2}.
\]

\item
The Tsallis entropy with parameter $\alpha>0$, $\alpha\neq 1$ of the lognormal distribution is given by the formula 
\[
\entropy_T^{LN} (\alpha,m,\sigma^2)
= \frac{1}{1-\alpha} \left(\sigma^{1-\alpha}(2\pi)^{\frac{1-\alpha}{2}} \alpha^{-\frac12}
\exp\left\{m(1-\alpha) + \frac{\sigma^2 (1-\alpha)^2 }{2\alpha}\right\} - 1 \right)
\]

\item
The generalized R\'{e}nyi entropy with parameters $\alpha,\beta>0, \alpha\neq \beta$ of the lognormal distribution is given by
\[
\entropy_{GR}^{LN}\left(\alpha,\beta,m,\sigma^2\right)
= \log\sigma + \frac12 \log(2\pi)  + m + \frac{\log\beta - \log\alpha}{2(\beta-\alpha)} + \frac{\sigma^2 (1 - \alpha\beta) }{2\alpha\beta}.
\]

\item
The Sharma--Mittal entropy with parameters $\alpha,\beta>0, \alpha\neq \beta, \alpha\neq 1, \beta\neq1$ of the lognormal distribution is given by 
\begin{align*}
\entropy_{SM}^{LN}\left (\alpha,\beta,m,\sigma^2\right )
= \frac{1}{1-\beta} \left(\sigma^{1-\beta}(2\pi)^{\frac{1-\beta}{2}} \alpha^{-\frac{1- \beta}{2(1 - \alpha)}}
\exp\left\{m(1-\beta) + \frac{\sigma^2 (1-\alpha)(1-\beta) }{2\alpha}\right\} - 1 \right).
\end{align*}

\item
The  Kullback--Leibler divergence (relative entropy) of two lognormal distributions is given by 
\[
\entropy_{KL}^{LN} \left(m,m_1,\sigma^2,\sigma_1^2\right) 
= \log\left(\frac{\sigma_1}{\sigma}\right) + \frac{\sigma^2 -\sigma_1^2 + (m-m_1)^2}{2\sigma_1^2}.
\]
\end{enumerate}
    
\end{proposition}

\begin{proof}
1) To calculate the Shannon entropy of $X_{m,\sigma^2}$, we express it in terms of the mean and variance of $\log X_{m,\sigma^2} \simeq \mathcal{N} (m,\sigma^2)$.
Substituting the expression for the lognormal density into the definition of the Shannon entropy, we get:
\begin{align*}
\entropy_{SH}^{LN}(m,\sigma^2)
&= -\int_0^\infty p_{m,\sigma^2}(x) \log \left( \frac{1}{x \sigma \sqrt{2\pi}} \exp\left\{- \frac{(\log x - m)^2}{2\sigma^2}\right\}\right) dx 
\\
&=  \log \left(\sigma \sqrt{2\pi}\right) \int_0^\infty p_{m,\sigma^2}(x)  dx 
+ \int_0^\infty p_{m,\sigma^2}(x) \log x  dx 
\\
&\quad + \int_0^\infty p_{m,\sigma^2}(x) \frac{(\log x - m)^2}{2\sigma^2} dx 
\\
&=  \log \left(\sigma \sqrt{2\pi}\right)
+ \ex \log X_{m,\sigma^2} + \frac{1}{2\sigma^2} \Var(\log X_{m,\sigma^2}) 
\\
&=  \log \left(\sigma \sqrt{2\pi}\right)
+ m + \frac{1}{2}.
\end{align*}

2) To compute the R\'enyi entropy $\entropy_{R}^{LN} (\alpha,m,\sigma^2)$, 
note that
\begin{align}
p_{m,\sigma^2}^\alpha(x) 
&= \frac{1}{\sigma^\alpha (2\pi)^{\alpha/2}}  \frac{1}{x^\alpha} \exp\left\{- \frac{\alpha(\log x - m)^2}{2\sigma^2}\right\}  
\notag\\
&=  \sigma^{1-\alpha}(2\pi)^{\frac{1-\alpha}{2}} \alpha^{-\frac12} x^{1-\alpha} p_{m,\sigma^2/\alpha}(x).
\label{lndensities}
\end{align}
Therefore,
\begin{align*}
\int_0^{\infty} p_{m,\sigma^2}^\alpha(x)dx 
&=  \sigma^{1-\alpha}(2\pi)^{\frac{1-\alpha}{2}} \alpha^{-\frac12} \int_0^{\infty} x^{1-\alpha} p_{m,\sigma^2/\alpha}(x) dx 
\\
&= \sigma^{1-\alpha}(2\pi)^{\frac{1-\alpha}{2}} \alpha^{-\frac12}\ex X_{m,\sigma^2/\alpha}^{1-\alpha},
\end{align*}
where $X_{m,\sigma^2/\alpha}$ denotes a lognormal random variable with parameters $m$ and $\sigma^2/\alpha$.
Using the moment formula \eqref{exln1} for a lognormal distribution, we have:
\begin{equation}\label{ln-int}
\int_0^{\infty} p_{m,\sigma^2}^\alpha(x)dx 
= \sigma^{1-\alpha}(2\pi)^{\frac{1-\alpha}{2}} \alpha^{-\frac12}
\exp\left\{m(1-\alpha) + \frac{\sigma^2 (1-\alpha)^2 }{2\alpha}\right\}.
\end{equation}
Hence, the R\'enyi entropy is given by
\begin{align*}
\entropy_{R}^{LN} (\alpha,m,\sigma^2) &=
\frac{1}{1-\alpha} \log\left(\int_0^{\infty} p_{m,\sigma^2}^\alpha(x)dx \right)
\\
& =\frac{1}{1-\alpha} \log\left(\sigma^{1-\alpha}(2\pi)^{\frac{1-\alpha}{2}} \alpha^{-\frac12}
\exp\left\{m(1-\alpha) + \frac{\sigma^2 (1-\alpha)^2 }{2\alpha}\right\}\right)
\\
& = \log\sigma + \frac12 \log(2\pi) + \frac{\log\alpha}{2(\alpha-1)}
+ m + \frac{\sigma^2 (1-\alpha)}{2\alpha}.
\end{align*}

3) Similarly to the proof of the statement 1), the integral from the definition of the generalized R\'enyi entropy can be decomposed as
\begin{align*}
\MoveEqLeft[1.5]
\int_0^\infty p_{m,\sigma^2}^\alpha(x) \log p_{m,\sigma^2}(x)dx
\\
&= - \log \left(\sigma \sqrt{2\pi}\right) \int_0^\infty p_{m,\sigma^2}^\alpha(x)  dx 
- \int_0^\infty p_{m,\sigma^2}^\alpha(x) \log x  dx 
- \frac{1}{2\sigma^2}\int_0^\infty p_{m,\sigma^2}^\alpha(x) (\log x - m)^2 dx
\\
&\eqqcolon R_1 + R_2 + R_3.
\end{align*}
Using \eqref{ln-int}, we immediately obtain
\[
R_1 = - \sigma^{1-\alpha}(2\pi)^{\frac{1-\alpha}{2}} \alpha^{-\frac12} \log \left(\sigma \sqrt{2\pi}\right) 
\exp\left\{m(1-\alpha) + \frac{\sigma^2 (1-\alpha)^2 }{2\alpha}\right\}.
\]
Further, using the representation \eqref{lndensities} and the identity \eqref{exln2}, we derive
\begin{align*}
R_2 &= - \sigma^{1-\alpha}(2\pi)^{\frac{1-\alpha}{2}} \alpha^{-\frac12} 
 \int_0^\infty x^{1-\alpha} (\log x) p_{m,\sigma^2/\alpha}(x)  dx
 \\
&= - \sigma^{1-\alpha}(2\pi)^{\frac{1-\alpha}{2}} \alpha^{-\frac12} 
 \ex\left[X_{m,\sigma^2/\alpha}^{1-\alpha} \log X_{m,\sigma^2/\alpha}\right]
  \\
&= - \sigma^{1-\alpha}(2\pi)^{\frac{1-\alpha}{2}} \alpha^{-\frac12} 
\left(\frac{\sigma^2(1-\alpha)}{\alpha} + m\right)\exp\left\{m(1-\alpha) + \frac{\sigma^2 (1-\alpha)^2 }{2\alpha}\right\}.
\end{align*}
Similarly, using \eqref{lndensities} and \eqref{exln3}, we get
\begin{align*}
R_3 &= - \sigma^{1-\alpha}(2\pi)^{\frac{1-\alpha}{2}} \alpha^{-\frac12} \frac{1}{2\sigma^2}\int_0^\infty  x^{1-\alpha}  (\log x - m)^2 p_{m,\sigma^2/\alpha}(x) dx
\\
&= - \sigma^{1-\alpha}(2\pi)^{\frac{1-\alpha}{2}} \alpha^{-\frac12} \frac{1}{2\sigma^2} \ex\left[X_{m,\sigma^2/\alpha}^{1-\alpha}  \left(\log X_{m,\sigma^2/\alpha} - m\right)^2\right]
\\
&= - \sigma^{1-\alpha}(2\pi)^{\frac{1-\alpha}{2}} \alpha^{-\frac12} \frac{1}{2\alpha} \left(\frac{\sigma^2 (1-\alpha)^2}{\alpha} +1\right)\exp\left\{m(1-\alpha) + \frac{\sigma^2 (1-\alpha)^2 }{2\alpha}\right\}.
\end{align*}
Thus, combining the above identities, we arrive at
\begin{multline}\label{ln-int2}
\int_0^\infty p_{m,\sigma^2}^\alpha(x) \log p_{m,\sigma^2}(x)dx
\\*
= - \sigma^{1-\alpha}(2\pi)^{\frac{1-\alpha}{2}} \alpha^{-\frac12} 
\left(\log \left(\sigma \sqrt{2\pi}\right) + \frac{\sigma^2\left(1-\alpha^2\right)}{2\alpha^2} + m + \frac{1}{2\alpha}\right)
\\*
{}\times\exp\left\{m(1-\alpha) + \frac{\sigma^2 (1-\alpha)^2 }{2\alpha}\right\}.
\end{multline}
Finally, inserting \eqref{ln-int2} and \eqref{ln-int} into the definition \eqref{eq:GR2dens} of the generalized R\'enyi entropy, we obtain
\begin{align*}
\entropy_{GR}^{LN} (\alpha,m,\sigma^2)
    &= - \frac{\int_0^\infty p_{m,\sigma^2}^\alpha(x) \log p_{m,\sigma^2}(x)dx}{\int_0^\infty p_{m,\sigma^2}^\alpha(x) dx}
\\*
&=\log \left(\sigma \sqrt{2\pi}\right) + \frac{\sigma^2\left(1-\alpha^2\right)}{2\alpha^2} + m + \frac{1}{2\alpha}.
\end{align*}

4) The formula for $\entropy_T^{LN} (\alpha,m,\sigma^2)$  follows by substituting the result \eqref{ln-int} into the definition \eqref{eq:Tdens} of the Tsallis entropy.

5) Substituting \eqref{ln-int} into \eqref{eq:GR1dens} and simplifying the resulting expressions, we obtain
\begin{align*}
\MoveEqLeft[0]
\entropy_{GR}^{LN}\left(\alpha, \beta, m, \sigma^2\right) 
= \frac{1}{\beta - \alpha}  \log\left(\frac{\int_0^\infty p_{m,\sigma^2}^\alpha(x) dx}{\int_0^\infty p_{m,\sigma^2}^\beta(x) dx}\right)
\\
&= \frac{1}{\beta - \alpha}  \log\left(\sigma^{\beta-\alpha}(2\pi)^{\frac{\beta-\alpha}{2}} \left(\frac{\beta}{\alpha}\right)^{\frac12}
\exp\left\{m(\beta-\alpha) + \frac{\sigma^2 (\beta-\alpha) (1 - \alpha\beta) }{2\alpha\beta}\right\}
\right)
\\
&= \log\sigma + \frac12 \log(2\pi) + \frac{\log\beta - \log\alpha}{2(\beta-\alpha)} + m + \frac{\sigma^2 (1 - \alpha\beta) }{2\alpha\beta}.
\end{align*}

6) $\entropy_{SM}(\alpha, \beta, m, \sigma^2)$ is computed by substituting formula \eqref{ln-int} into the definition \eqref{eq:SMdens} of the Sharma--Mittal entropy

7) The Kullback--Leibler divergence between two lognormal distributions with parameters $(m, \sigma^2)$ and $(m_1, \sigma_1^2)$ is given by
\begin{align*}
\MoveEqLeft
\entropy_{KL}^{LN} \left(m,m_1,\sigma^2,\sigma_1^2\right) 
= \int_0^\infty p_{m,\sigma^2}(x)\log\left(\frac{p_{m,\sigma^2}(x)}{p_{m_1,\sigma_1^2}(x)}\right) dx
\\
&=  \int_0^\infty p_{m,\sigma^2}(x)
\left(\log\left(\frac{\sigma_1}{\sigma}\right) - \frac{(\log x - m)^2}{2\sigma^2}
+ \frac{(\log x - m_1)^2}{2\sigma_1^2}\right) dx
\\
&= \log\left(\frac{\sigma_1}{\sigma}\right)
 -\frac{1}{2\sigma^2} \ex(\log X_{m,\sigma^2} - m)^2
+ \frac{1}{2\sigma_1^2}\ex(\log X_{m,\sigma^2} - m_1)^2.
\end{align*}
Noting that $\log X_{m,\sigma^2} \simeq \mathcal{N} (m,\sigma^2)$, we compute
\[
\ex(\log X_{m,\sigma^2} - m)^2 = \sigma^2
\quad\text{and}\quad
\ex(\log X_{m,\sigma^2} - m_1)^2
= \sigma^2 + (m-m_1)^2.
\]
Substituting these results completes the proof.
\end{proof}

\subsection{Examples of the calculation of the modified Shannon entropy  for the distributions with bounded density}\label{modify-ex}
In all examples we use equality \eqref{modif} from Proposition \ref{Prop1}. 
\begin{example}[Normal distribution]
    Let $p(x)=\frac{1}{\sigma\sqrt{2\pi}}e^{-\frac{x^2}{2\sigma^2}},\ \sigma>0$. Then $M=\frac{1}{\sigma\sqrt{2\pi}}$,
    $\entropy^{(norm)}_{SH}  =\frac12(1+\log 2\pi)+\log \sigma$, and
    $\entropy^{(norm)}_{SH, M}=\sigma\sqrt{\frac{\pi}{2}}$.
\end{example}
\begin{example}[Gamma distribution]\label{ex:mShga}
    Let $p(x)=\frac{\lambda^\mu}{\Gamma(\mu)}
    x^{\mu-1}e^{-\lambda x}\1_{x>0},\ \mu>0,\ \lambda>0$.
    Then for $\mu\in(0,1)\ p(x)$ is not bounded, and the modified Shannon entropy does not exist.
    Now, let $\mu>1$ (the case $\mu=1$ corresponds to
    Example~\ref{ex:mShexp} below).
    Then $p(x)$ has a unique maximum (in other words, it is a unimodal distribution) at point
    $x=\frac{\mu-1}{\lambda}$, and
    $$
    M=\frac{\lambda^\mu}{\Gamma(\mu)}
    \left(\frac{\mu-1}{\lambda}\right)^{\mu-1}e^{1-\mu}
    =\frac{\lambda}{\Gamma(\mu-1)}
    \left(\mu-1\right)^{\mu-2}e^{1-\mu}.
    $$
    Therefore,
    \begin{align*}
        \entropy_{SH,M}^{(gamma)}
        &=\frac{1}{M} \left(\entropy_{SH}^{(gamma)}+\log M\right)\\
        &=\frac{\Gamma(\mu-1)}{\lambda}(\mu-1)^{2-\mu}e^{\mu-1}
        (-\log\lambda+\log\Gamma(\mu)+\mu\\
        &\quad -\psi(\mu)(\mu-1)+\log\lambda-\log\Gamma(\mu-1)
        +(\mu-2)\log(\mu-1)+1-\mu)\\
        &=\frac{\Gamma(\mu-1)}{\lambda}(\mu-1)^{2-\mu}e^{\mu-1}
        (1+(\mu-1)\log(\mu-1)-\psi(\mu)(\mu-1)).
    \end{align*}
    We can use the recurrent relation for digamma function:
    $\psi(\mu)=\psi(\mu-1)+\frac{1}{\mu-1}$
    and reduce modified entropy to the form
   $$
    \entropy_{SH,M}^{(gamma)}
    =\frac{\Gamma(\mu-1)}{\lambda}(\mu-1)^{2-\mu}e^{\mu-1}
    (\mu-1)(\log(\mu-1)-\psi(\mu-1)),
   $$
   which, in addition, immediately implies that
   $\psi(z)<\log z$ for $z>0$.
   Finally, applying formula~8.362.2 from \cite{GradshteynRyzhik2007}
   (which itself states that
   $\psi(z)-\log z
   =-\sum_{k=0}^\infty \left[\frac{1}{z+k}-\log\left(1+\frac{1}{z+k}\right)
   \right]<0$ for $z>0$), we obtain that
   $$
    \entropy_{SH,M}^{(gamma)}
    =\frac{\Gamma(\mu)(\mu-1)^{2-\mu}e^{\mu-1}}{\lambda}
    \sum_{k=0}^\infty \left[\frac{1}{\mu-1+k}-\log\left(1+\frac{1}{\mu-1+k}\right)
   \right].
   $$    
\end{example}

\begin{example}[Exponential distribution]\label{ex:mShexp}
    Let $p(x)=\lambda e^{-\lambda x}\1_{x\ge 0}$. Then
    $M=\lambda$, and according to \eqref{eq:eshexp},
    $\entropy^{(exp)}_{SH}=1-\log \lambda$, and
    $\entropy^{(exp)}_{SH,M}=\frac{1}{\lambda}$.
\end{example}

\begin{example}[Chi-squared distribution]
Consider the chi-squared distribution with $\nu$ degrees of freedom.
Since it corresponds to a gamma distribution with parameters 
$\mu = \frac\nu2$ and $\lambda = \frac12$ 
(and reduces to the exponential distribution when $\nu = 2$),
we can apply the results from Examples \ref{ex:mShga}--\ref{ex:mShexp} to conclude the following:
\begin{itemize}
\item[$(i)$]
if $\nu = 1$, then the modified Shannon entropy does not exist;
\item[$(ii)$]
if $\nu = 2$, then $\entropy_{SH,M}^{(\chi^2)} = 2$;
\item[$(iii)$]
if $\nu \ge 3$, then 
\[
\entropy_{SH,M}^{(\chi^2)}
= 2 \Gamma\left(\frac\nu2\right) \left(\frac\nu2 - 1\right) ^{2-\frac\nu2}\exp\left\{\frac\nu2-1\right\} \left(\log\left(\frac\nu2-1\right) -\psi\left(\frac\nu2-1\right)\right).
\]
\end{itemize}
As discussed in Example \ref{ex:mShga}, the positivity of the final expression in $(iii)$ arises from the inequality $\psi(z) < \log z$ for $z>0$.
\end{example}
\begin{example}[Laplace  distribution] Let
    $p(x)=\frac{\lambda}{2} e^{-\lambda |x-\mu|},\,x\in\R$. Then $M=\frac{\lambda}{2}$, and $$\entropy_{SH, M}^{(La)}=-\frac{2}{\lambda}\log\left(\frac{\lambda}{2}\right)+\frac{2}{\lambda}+\frac{2}{\lambda}\log\left(\frac{\lambda}{2}\right)=\frac{2}{\lambda}.$$ 
     
\end{example}
\begin{example}[Log-normal  distribution]
Let
$p(x) = \frac{1}{x \sigma \sqrt{2\pi}} \exp\left\{- \frac{(\log x - m)^2}{2\sigma^2}\right\} \1_{x>0}$.
The function $p(x)$ achieves a unique maximum at the point
$x = \exp\{m-\sigma^2\}$
with the corresponding maximum value:
$M = \frac{1}{\sigma\sqrt{2\pi}}\exp\{\frac{\sigma^2}{2} - m\}$.
Then
\begin{align*}
\entropy_{SH,M}^{(LN)}
&= \frac{1}{M}\left(\entropy_{SH}^{(LN)}(m,\sigma^2) + \log M\right)
\\
&= \sigma\sqrt{2\pi}\exp\left \{m - \frac{\sigma^2}{2}\right \}
\left(\log\left(\sigma\sqrt{2\pi}\right)  + m + \frac12 
- \log\left(\sigma\sqrt{2\pi}\right) + \frac{\sigma^2}{2} - m \right)
\\
&= \sigma (\sigma^2+1) \sqrt{\frac\pi2}
\exp\left \{m - \frac{\sigma^2}{2}\right\}.
\end{align*}
\end{example}
\begin{example}[Uniform distribution]
   This distribution is degenerate on any interval $[a,b]$ in the sense of modified entropy because   its $\entropy^{(uni)}_{SH}=-\log\left(\frac{1}{b-a}\right)$, and $\entropy^{(uni)}_{SH, M}=0$.
\end{example}

\section{Dependence of entropies in parameters of distribution}\label{secdeppar}
In this section we consider only three distributions: Poisson, gamma and exponential. For Poisson distribution we give the new proof that its Shannon entropy increases with intensity, for gamma distribution we prove that its entropy decreases with parameter $\lambda$ and increases with parameter $\mu$, and for exponential distribution we study the behavior of all entropies because it can be achieved by similar way.  
\subsection{Poisson distribution}\label{poissparam}
Recall that a random variable $X^{(pois)}_\lambda$ has a Poisson distribution with parameter $\lambda > 0$ if
\[
\Prob\left(X^{(pois)}_\lambda = k\right) = \frac{\lambda^k  e^{-\lambda}}{k!},\quad k \ge 0.
\]

Its Shannon entropy $\entropy_{SH}^{(pois)} (\lambda)$ equals
\begin{equation}\label{entr1}
\entropy_{SH}^{(pois)} (\lambda) = -\sum_{k=0}^{\infty}\frac{\lambda^k e^{-\lambda}}{k!} \log \left(\frac{\lambda^k e^{-\lambda}}{k!}\right)=-\lambda \log \left(\frac{\lambda}{e}\right)+e^{-\lambda} \sum_{k=2}^{\infty} \frac{\lambda^k \log k!}{k!}.
\end{equation}

The following result was initially derived in \cite[Theorem 1]{braiman2024}. We provide a new proof here, followed by a discussion of the advantages of our approach in Remark~\ref{rem:comparison} below.

\begin{theorem}\label{prop:entropy-pois-incr}
The Shannon entropy $\entropy_{SH}^{(pois)}(\lambda)$ strictly increases as a function of~$\lambda$.
\end{theorem}

\begin{proof}
According to \eqref{entr1},
\[
\entropy_{SH}^{(pois)} (\lambda)
= e^{-\lambda} \sum_{i = 0}^{\infty} \frac{\lambda^i} {i!} \log i! - \lambda \log\lambda + \lambda.
\]
Note that $\sum_{i = 0}^{\infty} \frac{\lambda^i} {i!} \log i!$ is a power series, which converges for all $\lambda\in\R$; consequently it is infinitely differentiable.
The first derivative of $\entropy_{SH}^{(pois)}(\lambda)$ equals
\begin{align*}
\frac{d}{d\lambda}\entropy_{SH}^{(pois)}(\lambda)
&= - e^{-\lambda} \sum_{i = 0}^{\infty} \frac{\lambda^i} {i!} \log i!
+  e^{-\lambda} \sum_{i = 1}^{\infty} \frac{\lambda^{i-1}}{(i-1)!} \log i!
- \log\lambda
\\
&= e^{-\lambda} \sum_{i = 1}^{\infty} \frac{\lambda^i}{i!} \log (i+1)
- \log\lambda.
\end{align*}
We need to prove that $\frac{d}{d\lambda}\entropy_{SH}^{(pois)}(\lambda) > 0$ for all $\lambda>0$. To this end, we compute the second derivative:
\begin{align*}
\frac{d^2}{d\lambda^2}\entropy_{SH}^{(pois)}(\lambda)
&= - e^{-\lambda} \sum_{i = 1}^{\infty} \frac{\lambda^i}{i!} \log (i+1)
+ e^{-\lambda} \sum_{i = 1}^{\infty} \frac{\lambda^{i-1}}{(i-1)!} \log (i+1) - \frac1\lambda
\\
&= e^{-\lambda} \sum_{i = 1}^{\infty} \frac{\lambda^i}{i!} \log \frac{i+2}{i+1} +  e^{-\lambda} \log 2 - \frac1\lambda.
\end{align*}
Now, we apply the elementary inequality
$\log \frac{i+2}{i+1} = \log (1 + \frac{1}{i+1}) \le \frac{1}{i+1}$
and get
\begin{align*}
\frac{d^2}{d\lambda^2}\entropy_{SH}^{(pois)}(\lambda)
&\le e^{-\lambda} \sum_{i = 1}^{\infty} \frac{\lambda^i}{(i+1)!} +  e^{-\lambda} \log 2 - \frac1\lambda\\
&=  e^{-\lambda} \, \frac1\lambda \left(e^\lambda - 1 - \lambda \right) +  e^{-\lambda} \log 2 - \frac1\lambda
\\
&= -\frac1\lambda e^{-\lambda} + e^{-\lambda} \left(\log2 - 1\right) < 0.
\end{align*}
Thus, the first derivative $\frac{d}{d\lambda}\entropy_{SH}^{(pois)}(\lambda)$ strictly decreases.
Therefore, in order to prove its positivity, it suffices to show that
\begin{equation}\label{eq:lim-der-pos}
\lim_{\lambda\to\infty}\frac{d}{d\lambda}\entropy_{SH}^{(pois)}(\lambda) \ge 0.
\end{equation}
In fact, we can establish an even more precise result: this limit is exactly zero. This will be proven in Lemma~\ref{l:derHtozero} below.
\end{proof}

\begin{lemma}\label{l:derHtozero}
\begin{equation}\label{eq:derHtozero}
\lim_{\lambda\to\infty}\frac{d}{d\lambda}\entropy_{SH}^{(pois)}(\lambda) = 0.
\end{equation}
\end{lemma}

\begin{proof}
We start with proving the inequality \eqref{eq:lim-der-pos}.
We can represent $\frac{d}{d\lambda}\entropy_{SH}^{(pois)}(\lambda)$ in the following form
\[
\frac{d}{d\lambda}\entropy_{SH}^{(pois)}(\lambda)
= \log\lambda \left(\frac{e^{-\lambda}}{\log\lambda} \sum_{i = 1}^{\infty} \frac{\lambda^i}{i!} \log (i+1)
- 1\right) .
\]
Hence, to get \eqref{eq:lim-der-pos}, it suffices to prove that
\begin{equation}\label{eq:lim-der-low}
\lim_{\lambda\to\infty} \frac{e^{-\lambda}}{\log\lambda} \sum_{i = 1}^{\infty} \frac{\lambda^i}{i!} \log (i+1) \ge 1.
\end{equation}
According to Lemma \ref{l:series-to-inf} in Appendix, the function
$e^{-\lambda}\sum_{i = 1}^{\infty} \frac{\lambda^i}{i!} \log (i+1)$
tends to infinity as $\lambda\to\infty$.
Then applying the l'H\^opital rule and taking into account the possibility of term-by-term differentiation of the series, we get
\begin{align}
\MoveEqLeft
\lim_{\lambda\to\infty} \frac{e^{-\lambda}}{\log\lambda} \sum_{i = 1}^{\infty} \frac{\lambda^i}{i!} \log (i+1)
\notag\\
&= \lim_{\lambda\to\infty} \lambda \left(-e^{-\lambda} \sum_{i = 1}^{\infty} \frac{\lambda^i}{i!} \log (i+1)
+ e^{-\lambda} \sum_{i = 1}^{\infty} \frac{\lambda^{i-1}}{(i-1)!} \log (i+1)\right)
\notag\\
&= \lim_{\lambda\to\infty} \lambda e^{-\lambda} \left(-\sum_{i = 0}^{\infty} \frac{\lambda^i}{i!} \log (i+1)
+  \sum_{i = 0}^{\infty} \frac{\lambda^{i}}{i!} \log (i+2)\right)
\notag\\
&= \lim_{\lambda\to\infty} e^{-\lambda}
  \sum_{i = 0}^{\infty} \frac{\lambda^{i+1}}{i!} \log \left(1+\frac{1}{i+1}\right)
\ge \lim_{\lambda\to\infty} e^{-\lambda}
  \sum_{i = 0}^{\infty} \frac{\lambda^{i+1}}{i!} \left(\frac{1}{i+1} - \frac{1}{2(i+1)^2}\right)
\notag\\
&= \lim_{\lambda\to\infty} e^{-\lambda} \left(
  \sum_{i = 1}^{\infty} \frac{\lambda^{i}}{i!}
  - \frac12 \sum_{i = 1}^{\infty} \frac{\lambda^i}{i!\,i}\right),
\label{eq:lim-der-low1}
\end{align}
where we have used the elementary inequality
$\log (1 + x) \ge x - x^2/2$.
Obviously,
\begin{equation}\label{eq:lim-der-low2}
\lim_{\lambda\to\infty} e^{-\lambda}
  \sum_{i = 1}^{\infty} \frac{\lambda^{i}}{i!}
= \lim_{\lambda\to\infty} \frac{e^\lambda - 1}{e^\lambda}
= 1.
\end{equation}
Moreover, applying the l'H\^opital rule again and differentiating the series term-by-term, we obtain
\begin{equation}\label{eq:lim-der-low3}
\lim_{\lambda\to\infty}\frac{\sum_{i = 1}^{\infty} \frac{\lambda^i}{i!\,i}}{e^{\lambda}}
= \lim_{\lambda\to\infty}\frac{\sum_{i = 1}^{\infty} \frac{\lambda^{i-1}}{i!}}{e^{\lambda}}
= \lim_{\lambda\to\infty}\frac{e^\lambda - 1}{\lambda e^{\lambda}}
= 0.
\end{equation}
Combining \eqref{eq:lim-der-low1}--\eqref{eq:lim-der-low3} we get \eqref{eq:lim-der-low}, whence \eqref{eq:lim-der-pos} follows.

Now, using the Jensen inequality, we obtain the following upper bound in addition to the lower bound \eqref{eq:lim-der-pos}:
\begin{align*}
\lim_{\lambda\to\infty}\frac{d}{d\lambda}\entropy_{SH}^{(pois)}(\lambda)
&= \lim_{\lambda\to\infty} \left( e^{-\lambda} \sum_{i = 1}^{\infty} \frac{\lambda^i}{i!} \log (i+1)
- \log\lambda \right)\\
&= \lim_{\lambda\to\infty}\left( \ex\log \left(X^{(pois)}_\lambda + 1\right) - \log\lambda\right)
\\
&\le \lim_{\lambda\to\infty}\left( \log \left (\ex X^{(pois)}_\lambda + 1\right ) - \log\lambda\right)
=\lim_{\lambda\to\infty}\log \frac{\lambda+1}{\lambda} = 0.
\end{align*}
This concludes the proof \eqref{eq:derHtozero}.
\end{proof}

\begin{remark}\label{rem:comparison}
In comparison to the proof of Proposition \ref{prop:entropy-pois-incr} presented in \cite{braiman2024}, we note the following:
\begin{itemize} 
\item[$(i)$] In our Lemma~\ref{l:derHtozero}, we calculate the exact value of the limit, whereas in \cite{braiman2024}, the authors only establish a lower bound \eqref{eq:lim-der-pos}.
\item[$(ii)$] The inequality \eqref{eq:lim-der-low} was proved in \cite[Lemma A.1]{braiman2024} using a different approach, employing the two-sided Stirling approximation for factorials. In contrast, our proof -- based on series differentiation and elementary inequalities -- is significantly more concise.
\end{itemize}
\end{remark}

 
\subsection{Gamma distribution}
\begin{theorem}
    The Shannon entropy of gamma distribution decreases in $\lambda>0$ from $+\infty$ to $-\infty$, for any $\mu>0$, and increases from $-\infty$ to $+\infty$ in $\mu>0$ for any $\lambda>0$.
\end{theorem}
\begin{proof}
    According to equality \eqref{eq:gamsh},
    $\entropy_{SH}^{(gamma)}(\lambda,\mu)$ decreases in $\lambda>0$.
    Its behaviour in $\mu$ is more involved.

Denote $SH(\mu)=\log\Gamma(\mu)+\mu-\psi(\mu)(\mu-1)$ and differentiate it:
\begin{equation*}
    (SH(\mu))'=\psi(\mu)+1-\psi'(\mu)(\mu-1)-\psi(\mu)
    =1-\psi'(\mu)(\mu-1).
\end{equation*}
It follows, for example, from the representation
\begin{equation*}
   \psi'(z) =\sum_{n=0}^{\infty}\frac{1}{(z+n)^2},\ z>0,
\end{equation*}
that $\psi'(\mu)>0$ and decreases in $\mu$. For $\mu<1$, obviously,
$1-\psi'(\mu)(\mu-1)>0$, therefore,
$\entropy_{SH}^{(gamma)}(\lambda,\mu)$ increases in $\mu\in(0,1)$.

Now, let $\mu>1$. Applying recurrent relation
$\psi'(z+1) =\psi'(z) -\frac{1}{z^2}$, we come to the formula
$$
1-\left(\psi'(\mu-1)-\frac{1}{(\mu-1)^2}\right)(\mu-1), \ \mu>1,
$$
or, that is the same,
$$
1-\left(\psi'(\mu)-\frac{1}{\mu^2}\right)\mu, \ \mu>0.
$$
 In other words, we need to consider the sign of the value 
 $$
 \frac{1}{\mu}+\frac{1}{\mu^2}-\psi'(\mu),\ \mu>0.
 $$
However, it is evident that
\begin{align*}
    \psi'(\mu)&=\frac{1}{\mu^2}+\frac{1}{(\mu+1)^2}+\frac{1}{(\mu+2)^2}+\ldots
    < \frac{1}{\mu^2}+\int_{\mu}^{\mu+1}\frac{dx}{x^2}
    +\int_{\mu+1}^{\mu+2}\frac{dx}{x^2}  +\ldots\\
    &=\frac{1}{\mu^2}+\int_{\mu}^{\infty}\frac{dx}{x^2}
    =\frac{1}{\mu^2}+\frac{1}{\mu},
\end{align*}
therefore,
$\frac{1}{\mu^2}+\frac{1}{\mu}-\psi'(\mu)>0,\ \mu>0$,
and it means that for any $\lambda>0$ the entropy
$\entropy_{SH}^{(gamma)}(\lambda,\mu)$ increases in $\mu>0$.

Let us calculate
\begin{align*}
\lim_{\mu \downarrow 0} SH(\mu)
&= \lim_{\mu \downarrow 0} \bigl(\log\Gamma(\mu) - \psi(\mu)(\mu-1)\bigr)
\\
&= \lim_{\mu \downarrow 0} \left (\log\frac{\Gamma(\mu+1)}{\mu} - \left( \psi(\mu + 1) - \frac1\mu\right)(\mu-1)\right)
\\
&= \lim_{\mu \downarrow 0} \left (\log\frac{1}{\mu}+ \psi(1) + 1 - \frac1\mu\right) = -\infty.
\end{align*}
Here we have used the recurrent equality
$\psi(z+1) = \psi(z) + \frac1z$.

Furthermore, using asymptotic relations for $\Gamma(\mu)$ and $\psi(\mu)$, we obtain
\begin{align*}
\lim_{\mu \to\infty} SH(\mu)
&= \lim_{\mu \to \infty} \bigl(\log\Gamma(\mu) +\mu - \psi(\mu)(\mu-1)\bigr)
\\
&\ge \lim_{\mu \to \infty} \left (\log\left(\sqrt{2\pi\mu}\left(\frac{\mu}{e}\right)^\mu\right) + \mu - \left( \log\mu - \frac{1}{2\mu}\right)(\mu-1)\right)
\\
&= \lim_{\mu \to \infty} \left (\frac12\log(2\pi\mu) + \mu\log\mu -  \mu\log\mu  +\frac12 + \log\mu - \frac{1}{2\mu}\right)
\\
&= \lim_{\mu \to \infty} \left (\frac12\log(2\pi\mu) +\frac12 + \log\mu - \frac{1}{2\mu}\right)
= +\infty. \qedhere
\end{align*}
\end{proof}

\subsection{Exponential  distribution}
Now, let us analyze some properties of the entropies of exponential distribution. 
\begin{theorem}\label{prop:entropy-exp-decr}
\begin{enumerate}[1)]
\item The Shannon entropy $\entropy_{SH}^{(exp)}(\lambda)$,
the R\'{e}nyi entropy $\entropy_{R}^{(exp)}(\alpha,\lambda)$ with parameter $\alpha>0$,
the generalized R\'{e}nyi entropy $\entropy_{GR}^{(exp)}(\alpha,\lambda)$ with parameter $\alpha>0$ and
the generalized R\'{e}nyi entropy $\entropy_{GR}^{(exp)}(\alpha,\beta,\lambda)$ with parameters
$\alpha,\beta>0$, $\alpha\neq \beta$,
strictly decrease in~$\lambda\in(0,+\infty)$ from $+\infty$ to $-\infty$,
and
\begin{equation}\label{eq:expshelims}
    \entropy_{SH}^{(exp)}(e)
    =\entropy_{R}^{(exp)}(\alpha,\alpha^{\frac{1}{\alpha-1}})
    =\entropy_{GR}^{(exp)}(\alpha,e^{\frac{1}{\alpha}})
    =\entropy_{GR}^{(exp)}\left(\alpha,\beta,\left(\frac{\beta}{\alpha}\right)^{\frac{1}{\beta-\alpha}}\right)
    =0.
\end{equation}

\item For any $\lambda>0$ we have the following relations: if $\alpha<1$, then $\entropy_{R}^{(exp)}(\alpha,\lambda)>\entropy_{SH}^{(exp)}(\lambda)$,  while for  $\alpha>1$ $\entropy_{R}^{(exp)}(\alpha,\lambda)<\entropy_{SH}^{(exp)}(\lambda)$. Moreover, $$\lim_{\alpha\to 1}\entropy_{R}^{(exp)}(\alpha,\lambda)=\entropy_{SH}^{(exp)}(\lambda),$$
for any fixed $\lambda>0$ the R\'{e}nyi entropy decreases in $\alpha\in(0,+\infty)$ from $+\infty$ to $-\log\lambda.$

\item For any $0<\alpha<1$ the Tsallis entropy $\entropy_{T}^{(exp)}(\alpha,\lambda)$ strictly decreases in $\lambda~\in~(0,+\infty)\ $ from $\ +\infty\ $ to $\ \frac{1}{\alpha-1}<0$ and for any $\alpha>1\ $ $\entropy_{T}^{(exp)}(\alpha,\lambda)$ strictly decrease in~$\lambda\in(0,+\infty)$ from $\frac{1}{\alpha-1}>0$ to $-\infty$, and for any $\alpha>0$, $\alpha\neq 1$ we have
\begin{equation}
    \entropy_{T}^{(exp)}(\alpha,\alpha^{\frac{1}{\alpha-1}}))=0.
    \label{eq:etexplim}
\end{equation}

\item For any $\beta<1$ Sharma--Mittal entropy $\entropy_{SM}^{(exp)}(\alpha,\beta,\lambda)$ with parameters $\alpha,\beta>0, \alpha\neq \beta, \alpha\neq 1, \beta\neq1$ strictly decrease in $\lambda~\in~(0,+\infty)\ $ from $\ +\infty\ $ to $\ \frac{1}{\beta-1}<0$ and for any $\beta>1\ $ $\entropy_{SM}^{(exp)}(\alpha,\beta,\lambda)$ strictly decrease in~$\lambda\in(0,+\infty)$ from $\frac{1}{\beta-1}>0$ to $-\infty$, and for any $\alpha,\beta>0, \alpha\neq \beta, \alpha\neq 1, \beta\neq1$ we have
\begin{equation}
    \entropy_{SM}^{(exp)}(\alpha,\beta,\alpha^{\frac{1}{\alpha-1}}))=0.
    \label{eq:esmexplim}
\end{equation}
\end{enumerate}
\end{theorem}
\begin{proof}
1) All derivatives in $\lambda$   equal $-\frac{1}{\lambda}$ and are strictly negative for any $\lambda>0$, consequently, respective entropies  strictly decrease in parameter $\lambda>0$.
The equalities \eqref{eq:expshelims} are obvious.

2) For $\alpha<1$ we have that $\frac{\log\alpha}{\alpha-1}>1$ while for  $\alpha>1$, oppositely,  $\frac{\log\alpha}{\alpha-1}<1$. Moreover, $$\lim_{\alpha\to 1}\frac{\log\alpha}{\alpha-1}=1,$$ and it is very easy to check that the function $\frac{\log\alpha}{\alpha-1}, \alpha\in(0, +\infty)$ decreases from $+\infty$ to $0$ whence the proof follows.

3), 4) Decreasing of the Tsallis and Sharma--Mittal entropies follows directly from the relations \eqref{eq:etexp} and \eqref{eq:esmexp2}. The equalities \eqref{eq:etexplim} and \eqref{eq:esmexplim} are also obvious.
\end{proof}

\section{Convergence  of Shannon entropies}\label{robust}
In this section we solve a following problem. It is known that under some conditions the sequence of discrete probability distributions converges to  another probability distribution. What will be the behaviour of the respective Shannon entropies? In the case of convergence, we can say that in this sense, Shannon entropy has the property of robustness. 
\subsection{Negative binomial and logarithmic distribution}

\begin{definition}
   A random variable $X_{p,r}^{(nb)}$ has negative binomial distribution with parameters $p$ and $r$, if $X_{p,r}^{(nb)}$ attains values $k\in \mathbb{N}\cup \{0\}$, and
   \begin{gather*}
   \Prob\{X_{p,r}^{(nb)}=k\}=\frac{\Gamma(k+r)}{k!\Gamma(r)}(1-p)^kp^r,
   \end{gather*}
   where $p\in (0,1)$ and $r>0$, $\Gamma(\cdot)$ is a Gamma function.
\end{definition}

\begin{definition}
A random variable $X_{p}^{(log)}$ has logarithmic distribution with parameter $p\in (0,1)$, if
 \begin{gather*}
   \Prob\{X_{p}^{(log)}=k\}=-\frac{(1-p)^k}{k\log p},\ k\in\mathbb{N}.
   \end{gather*}
\end{definition}

The next fact is very well-known and easy to prove.

\begin{proposition}\label{lognblim}
    Logarithmic distribution is the limit of conditional negative binomial distribution in the following sense: for any $k\ge 1$
    \begin{gather*}
        \lim_{r\to 0}\Prob\{X_{p,r}^{(nb)}=k \,|\, X_{p,r}^{(nb)}>0\}=
        \lim_{r\to 0}\frac{\frac{\Gamma(k+r)}{k!\Gamma(r)}(1-p)^kp^r}{1-p^r}
        =-\frac{(1-p)^k}{k\log p}.
    \end{gather*}
\end{proposition}
For more connections between negative binomial distribution and logarithmic distribution see e.g.,
\cite{Anscombe1950, Mauricio2008}.

Taking Proposition \ref{lognblim} into account, we shall consider ``conditional'' Shannon entropy of the negative binomial distribution:
\begin{gather}\label{entrnb}
    \mathbf{H}_{SH}(X_{p,r}^{(nb)})=-\sum_{k=1}^\infty \mathbf{P}_{p,r}(k)\log \mathbf{P}_{p,r}(k)=\sum_{k=1}^\infty \mathbf{P}_{p,r}(k)\log \frac{1}{\mathbf{P}_{p,r}(k)},
\end{gather}
where $\mathbf{P}_{p,r}(k)=\Prob\{X_{p,r}^{(nb)}=k\,|\, X_{p,r}^{(nb)}>0\}$.

\begin{proposition} Conditional entropies converge:
    $$\lim_{r\to 0} \mathbf{H}_{SH}(X_{p,r}^{(nb)})= \mathbf{H}_{SH}(X_{p}^{(log)}),$$ where $\mathbf{H}_{SH}(X_{p}^{(log)})=\sum_{k=1}^\infty \mathbf{P}_{p}(k)\log \mathbf{P}_{p}(k)$,
    $\mathbf{P}_{p}(k)=-\frac{(1-p)^k}{k\log p},\ k\in \mathbb{N}$.
\end{proposition}

\begin{proof}
    Taking into account Proposition~\ref{lognblim}, it is sufficient to find a summable dominant. It follows from~\eqref{entrnb} that we need in both lower and upper bound for $\mathbf{P}_{p,r}(k)$. Since
    $\sqrt{p}<p^r<1$ for $r<\frac{1}{2}$, it is sufficient to bound
    $\frac{\Gamma(k+r)}{\Gamma(r)(1-p^r)}$ from below and from above.
    Note that $\lim_{r\to 0} r\Gamma(r)=1$, so we can consider such $r$ that $\frac{1}{2}<r\Gamma(r)<2$.
    Moreover, $\lim_{r\to 0} \frac{r}{1-p^r}=-\frac{1}{\log p}$, and we can assume that
    \[
    \frac{1}{4\log \frac{1}{p}}<\frac{1}{\Gamma(r)(1-p^r)}<\frac{4}{\log\frac{1}{p}}.\qedhere
    \]
\end{proof}

Also, for $0<r<1\ \Gamma(k)<\Gamma(k+r)<\Gamma(k+1)$. Finally, for such $r$
\begin{gather*}
    \mathbf{P}_{p,r}(k)\log \frac{1}{\mathbf{P}_{p,r}(k)}\le
    \frac{\Gamma(k+1)}{k!}(1-p)^k\left[\log{k} +\log{\left(\frac{4\log\frac{1}{p}}{\sqrt{p}}\right)}+k\log{\frac{1}{1-p}}\right]\\=(1-p)^k\left[\log{k} +\log{\left(\frac{4\log\frac{1}{p}}{\sqrt{p}}\right)}+k\log{\frac{1}{1-p}}\right],
\end{gather*}
and the series
\begin{gather*}
   \frac{4}{\log\frac{1}{p}}\sum_{k=1}^{\infty}(1-p)^k
   \left[\log{k} +\log{\left(\frac{4\log\frac{1}{p}}{\sqrt{p}}\right)}+k\log{\frac{1}{1-p}}\right]
\end{gather*}
converges, whence the proof follows.

\subsection{Binomial and Poisson distribution}\label{binpoiss}
\begin{definition}
    A random variable $X_{p,n}^{(bin)}$ has binomial distribution with parameter $p\in (0,1)$,
    if $X_{p,n}^{(bin)}\in\{0,1,\ldots,n\}$ and
    \begin{gather*}
        \Prob \left \{X_{p,n}^{(bin)}=k\right\}=\binom{n}{k}p^k(1-p)^{n-k},\ 0\le k\le n.
    \end{gather*}    
\end{definition}

Consider the sequence of random variables $X_{p_n,n}^{(bin)},\ n\ge 1$, such that $np_n\to\lambda>0$ as $n\to\infty$.
It is well known that in this case $X_{p_n,n}^{(bin)}$ converges in distribution to Poisson random variable $X_\lambda^{(pois)}$, that is
\begin{gather*}
    \Prob_{p_n}(k)=\binom{n}{k} p_n^k(1-p_n)^{n-k}\to
    e^{-\lambda}\frac{\lambda^k}{k!},\ n\to \infty,
    \ 0\le k\le n.
\end{gather*}

Now we prove that the sequence of the respective Shannon entropies of binomial distributions   converges to the entropy of Poisson distribution.

\begin{proposition}
$\displaystyle    \entropy_{SH}^{(bin)}(p_n,n)=\sum_{k=1}^n
    \Prob _{p_n}(k)\log \left(\frac{1}{\Prob _{p_n}(k)}\right)
    \to \entropy_{SH}^{(pois)}(\lambda)
$,
as $n\to\infty$.
\end{proposition}
\begin{proof}
As before, we shall bound $\Prob _{p_n}(k)$ from below and from above in order to find summable dominant.
Consider two cases.

(i) Let $0\le k \le \lfloor{\frac{n}{2}}\rfloor$.
Then, on one hand,
\begin{gather}
    \Prob_{p_n}(k)=\frac{n!}{k!(n-k)!}p_n^k(1-p_n)^{n-k}
    =\frac{(n-k+1)\cdot\ldots\cdot n}{k!}p_n^k(1-p_n)^{n-k}
    \nonumber\\
    =\frac{(1-\frac{k-1}{n})\cdot\ldots\cdot (1-\frac{1}{n})}{k!}(np_n)^k(1-p_n)^{n-k}
    \le \frac{1}{k!}(np_n)^k(1-p_n)^{n-k}.
    \label{poisprobest}
\end{gather}
Starting from some $n_0$ we have that $np_n\le 2\lambda$. So, let us consider  $n\ge n_0$,  and  fix that   $np_n\le 2\lambda$.
Then for $n\ge n_0$ $\Prob_{p_n}(k)\le \frac{(2\lambda)^k}{k!}$.
Note that  this upper bound is valid for all $0\le k \le n$, not only for  $0\le k \le \lfloor{\frac{n}{2}}\rfloor$.

On the other hand, for $n\ge 2$ and $0\le k\le \lfloor{\frac{n}{2}}\rfloor$, it follows from the equalities in \eqref{poisprobest} that $\Prob_{p_n}(k)\ge \dfrac{(\frac12)^{k-1}}{k!}(np_n)^k(1-p_n)^{n-k}$.

Furthermore, starting from some $n_1$ we have that $np_n\ge \frac{\lambda}{2}$. So, let us consider   $n\ge n_0\vee n_1$, so that
$\frac{\lambda}{2}\le np_n\le 2\lambda$.
Then $(np_n)^k\ge (\frac{\lambda}{2})^k$, and
$(1-p_n)^{n-k}\ge (1-\frac{2\lambda}{n})^{n}
=(1-\frac{2\lambda}{n})^{\frac{n}{2\lambda}2\lambda}$
(assuming without loss of generality that $\frac{2\lambda}{n} < 1$).

Note that starting from some $n_2$  it holds that 
$(1-\frac{2\lambda}{n})^{\frac{n}{2\lambda}}
\ge \frac{1}{2e}$, because $(1-\frac{2\lambda}{n})^{\frac{n}{2\lambda}}\to e^{-1}$, $n\to \infty$.
Then for   $n\ge n_0\vee n_1\vee n_2$ and for $0 \le k\le \lfloor{\frac{n}{2}}\rfloor$ we have that
$\Prob_{p_n}(k)\ge \frac{1}{2^kk!}
(\frac{\lambda}{2})^ke^{-2\lambda}
$, and
\begin{gather*}
    \Prob_{p_n}(k)\log \frac{1}{\Prob_{p_n}(k)}
    \le \frac{(2\lambda)^k}{k!}
\log\left(\frac{k!4^k}{\lambda^k}e^{2\lambda}\right)
= \frac{(2\lambda)^k}{k!}
\left(\log{k!}+k\log{\frac{4}{\lambda}}+
2\lambda\right),
\end{gather*}
and all series
$\sum_{k=0}^\infty \frac{(2\lambda)^k}{k!}\log{k!}$, 
$\sum_{k=0}^\infty \frac{(2\lambda)^k}{(k-1)!}$
and $\sum_{k=0}^\infty \frac{(2\lambda)^k}{k!}$ converge.

Therefore, we can put $a_k^{(n)}=0$ for $k>\lfloor\frac n2 \rfloor$,
$a_k^{(n)}=-\Prob_{p_n}(k)\log \Prob_{p_n}(k)$ for $0\le k\le \lfloor\frac n2 \rfloor$
and get that
\begin{gather*}
\sum_{k=0}^\infty a_k^{(n)}\to
-\sum_{k=0}^\infty e^{-\lambda}\frac{\lambda^k}{k!}\log\left(e^{-\lambda}\frac{\lambda^k}{k!}\right),\ n\to\infty, 
\shortintertext{i.e.}
\sum_{k=0}^{\lfloor\frac n2 \rfloor}
\Prob_{p_n}(k)\log {\frac{1}{\Prob_{p_n}(k)}}\to
\entropy_{SH}^{(pois)}(\lambda),\ n\to\infty.
\end{gather*}

Now it is sufficient to prove that the sum
$$
S_n=\sum_{k=\lfloor\frac{n}{2} \rfloor+1}^n
\Prob_{p_n}(k)\log {\frac{1}{\Prob_{p_n}(k)}}\to 0,
\ n\to\infty.
$$
So, consider the case

(ii) Let $\lfloor\frac{n}{2}\rfloor<k\le n$.
As before, $\Prob_{p_n}(k)\le \frac{(2\lambda)^k}{k!}$.
Since $\frac{(2\lambda)^k}{k!}\to 0,\ k\to\infty$, and moreover
$\frac{(2\lambda)^k}{k!}$ decreases in $k$ if $k>2\lambda$,
consider such $n\ge n_0\vee n_1\vee n_2$ that $\lfloor\frac{n}{2}\rfloor> 2\lambda$. For these $n$ we have the upper bound
$$\frac{(2\lambda)^k}{k!}\le
\frac{(2\lambda)^{\lfloor\frac{n}{2}\rfloor+1}}{\left(\lfloor\frac{n}{2}\rfloor+1\right)!}
$$ for $k\ge \lfloor\frac{n}{2}\rfloor.$ 

Now, consider function $f(x)=x\log{\frac{1}{x}},\ x\in (0,1)$.
It increases when $x$ increases from 0 to $\frac1e$ and then  decreases. Therefore, if $x_0\le \frac1e$ and $0<x<y\le \frac1e$, then $f(x)<f(y)$. It means that for $k\ge \lfloor\frac{n}{2}\rfloor$ and such $n$ from the set of numbers chosen according to all above assumptions and additionally such that  
$
\frac{(2\lambda)^{\lfloor\frac{n}{2}\rfloor+1}}{\left(\lfloor\frac{n}{2}\rfloor+1\right)!}<\frac1e
$, it holds that 
$$
\Prob_{p_n}(k)\log {\frac{1}{\Prob_{p_n}(k)}}
\le \frac{(2\lambda)^{\lfloor\frac{n}{2}\rfloor+1}}{\left(\lfloor\frac{n}{2}\rfloor+1\right)!}
\log{\frac{\left(\lfloor\frac{n}{2}\rfloor+1\right)!}{(2\lambda)^{\lfloor\frac{n}{2}\rfloor+1}}}
\le \frac{(2\lambda)^{n+1}}{\left(n+1\right)!}
\log{ \frac{\left(n+1\right)!}{(2\lambda)^{n+1}}},
$$
where the last inequality holds because we consider such $n$ that $\lfloor\frac{n}{2}\rfloor>2\lambda$ and for them
$\frac{(2\lambda)^{\lfloor\frac{n}{2}\rfloor+1}}{\left(\lfloor\frac{n}{2}\rfloor+1\right)!}
<\frac{(2\lambda)^{n+1}}{\left(n+1\right)!}
<\frac1e$.
It means that for such $n$ the sum
\begin{gather*}
    S_n\le n \frac{(2\lambda)^{n+1}}{\left(n+1\right)!}
    \left[\log{(n+1)!-(n+1)\log{2\lambda}}\right]\\
    \le n\frac{(2\lambda)^{n+1}}{\left(n+1\right)!}
    (n+1)\log{(n+1)}
    +\frac{(2\lambda)^{n+1}}{\left(n-1\right)!}
    |\log{2\lambda}|\\
    =\frac{(2\lambda)^{n+1}}{\left(n-1\right)!}
    \left(\log{(n+1)}+|\log{2\lambda}|\right)
    \to 0,\ n\to \infty,
\end{gather*}
and the proof follows.
\end{proof}

\section{Extreme values of Shannon entropy   of Gaussian vectors}\label{sec6}
In this section we will study the following question. Consider centered Gaussian random variable  with known variance $\sigma^2$. In the paper \cite{MMRR} we provided the values of various entropies  of this random variable and stated that all these entropies  increase  with $\sigma^2$.  In particular, Shannon entropy is calculated by the formula   \eqref{eq:SHdens}  and equals $$\mathbf{H}^{norm}_{SH}=\frac12(1+\log(2\pi))+\frac12\log\sigma^2.$$ Furthermore, if we consider a Gaussian stochastic process, for example, Wiener process or fractional Brownian motion, for which we know their one-dimensional distributions, we know their variances at each moment of time and therefore know the behavior  of the corresponding Shannon entropy in time. This issue was also studied in \cite{MMRR}. However, it is certain that the behavior of the entropy of a Gaussian process is not limited to its one-dimensional distributions and is determined by the family of finite-dimensional distributions, i.e., in our case, by the respective Gaussian vectors.  The formula for the Shannon entropy for Gaussian vector $\xi=(\xi, \xi_2, \ldots,\xi_n)$ with covariance matrix $A_n$ was provided in \cite{stratonovich} and has a form  $$\mathbf{H}^{norm}_{SH}(\xi)=\frac{n}{2}(1+\log(2\pi))+\frac{1}{2}\log(\det(A_n)),$$ where $A_n$ is a covariance matrix of vector $\xi$. Obviously, the main role in this formula is played by $\det(A_n)$, which in fact is the determinant of    positive-semidefinite matrix. In the general case, calculation of this determinant analytically is impossible. For example, we do not know any closed formula for its value in the case of fractional Gaussian noise, because of the complicated form of covariances. This problem was studied in detail,    e.g, in \cite{MMRS}.  Therefore, let us consider related  problems that admit analytical solutions. Namely, returning to the one-dimensional distributions, it is natural to ask the following question: if we have a Gaussian vector with fixed variances, what are extreme values of its Shannon entropy? Using algebraic language: if we have a positive-semidefinite matrix with fixed diagonal elements, what are the extreme values of its determinant? And, returning to Gaussian vectors, on which of them are extreme values  realized, if such vectors exist?  In this connection, we provide Hadamard's theorem on determinants for positive-semidefinite matrices.
\subsection{Hadamard's theorem on determinants for positive-semidefinite matrices} Consider quadratic  matrix $A_{n}=(a_{ij})_{i,j=1}^n$ and assume that it is positive semidefinite, and its 
diagonal elements $a_{ii}>0, \forall\ i=\overline{1,n},$ are fixed. Let us recall the Hadamard’s inequality for the determinants of positive semidefinite matrices, see e.g., \cite{Bellman1997, Hadamard1893, ROZANSKI2017} or \cite[p. 1077]{GradshteynRyzhik2007}. 
 
\begin{theorem}[Hadamard's inequality {\cite[Theorem 11]{ROZANSKI2017}}] 
    For the positive semidefinite Hermitian matrix $A_{n}=(a_{ij})_{i,j=1}^n,\ n\ge 2,$ with non-zero elements in the diagonal, the inequality
    \begin{gather*}
    \det(A_n)\le \prod_{i=1}^{n}a_{ii},
    \end{gather*}
    is satisfied.
    Equality holds if and only if the matrix $A_{n}$ is a diagonal matrix.
\end{theorem}
\begin{remark} Obviously, minimal value of determinant is zero. Now, let us consider  these maximal and minimal values in application to covariance matrices. 
   In this connection, consider  the following cases.
 \end{remark}
\begin{enumerate}
 \item Let $a_{ii}>0, 1\le i\le n$, and $a_{ij}=0,\ i\neq j$. Thus we have diagonal matrix, that is,  $\det(A_{n})=\prod_{k=1}^na_{kk}$. This case corresponds to  the  non-degenerate  Gaussian vector  with non-correlated consequently independent components.
 
 \item Let $a_{ij}=\sqrt{a_{ii}a_{jj}},\ i\neq j$. In this case all leading principal
minors equal zero, and determinant equals
 \begin{gather*}
 \det(A_{n}) =\det(\{\sqrt{a_{ii}a_{jj}}\})_{i,j=1}^n
\\=\left(\prod_{k=1}^na_{kk}\right)^{\frac12}
\begin{vmatrix}
 \sqrt{a_{11}} & \sqrt{a_{11}} & \ldots &   \sqrt{a_{11}}\\
 \sqrt{a_{22}} & \sqrt{a_{22}} & \ldots &   \sqrt{a_{22}}\\
 \ldots & \ldots &  \ddots &  \ldots\\ 
 \sqrt{a_{nn}} & \sqrt{a_{nn}} & \ldots &  \sqrt{a_{nn}}
\end{vmatrix}=
\left(\prod_{k=1}^na_{kk}\right)
\begin{vmatrix}
 1 & 1 & \ldots   & 1\\
 1 & 1 & \ldots   & 1\\
 \ldots & \ldots & \ddots &    \ldots\\ 
 1 & 1 & \ldots  & 1
\end{vmatrix}=0.
\end{gather*}
\end{enumerate}
  How can we realize in terms of a Gaussian vector the situation when this minimal value is reached? Let $\xi_0\simeq \mathcal{N}(0,1)$ be a standard Gaussian variable, and
$\xi_i=\sqrt{a_{ii}}\xi_0,\ 1\le i\le n$.
Then $\ex\xi_i=0,\ \ex\xi_i^2=a_{ii},\ \ex\xi_i\xi_j=\sqrt{a_{ii}}\sqrt{a_{jj}}$ and Gaussian vector
$ {\xi}=(\xi_1,\xi_2,\ldots,\xi_n)=(\sqrt{a_{11}}\xi_0,\sqrt{a_{22}}\xi_0, \ldots,\sqrt{a_{nn}} \xi_0 )$ indeed has the   covariance matrix
\begin{gather*}
 A_n=\{\sqrt{a_{ii}a_{jj}}\}_{i,j=1}^n.
  \end{gather*} 
Of course, such vector and its covariance matrix with zero determinant are not unique. Indeed,  we can consider the vector $${\xi}'=(\pm\sqrt{a_{11}}\xi_0,\pm\sqrt{a_{22}}\xi_0, \ldots,\pm\sqrt{a_{nn}} \xi_0 )$$ with any combination of signs, and the result will be the same.

\begin{remark} As an example, consider fractional Brownian motion (fBm) and fractional Gaussian noise (fGn). Fractional Brownian motion  $B^H=\{B^H_t,\ t\ge 0\}$ with Hurst index $H\in(0,1)$ is a centered Gaussian process with covariance function
\begin{gather*}
\ex B_t^HB_s^H=\frac12\left(t^{2H}+s^{2H}-|t-s|^{2H}\right).
\end{gather*}
\end{remark}

It is possible to consider fBm also for $H=0$ and $H-1$.
In this cases we put, respectively,  $B_t^1=t\xi$,
where $\xi\sim \mathcal{N}(0,1)$, and $B_t^0=\xi_t-\xi_0$, where all $\xi_t,\ t\ge 0$, are independent $\mathcal{N}(0,\frac12)$-random variables.

Now, fGn is created as a stochastic process with discrete time of the form $$\{\Delta_k^H=B_k^H-B_{k-1}^H,\ k \in\mathbb{N} \}.$$ It is a   stationary process. Its covariance has a following form: for any $k>l$
\begin{gather*}
  \ex \Delta_k^H \Delta_l^H= \ex \Delta_{k-l+1}^H \Delta_1^H
  =\frac12\left((k-l+1)^{2H}-2(k-l)^{2H}+(k-l-1)^{2H}\right),
\end{gather*}
see, e.g., \cite{MMRS}. For any $k\ge 1\ \ex(\Delta_k^H)^2=1$.
Taking successive values  of fGn, we can  create    Gaussian vector
 \begin{gather*}
     \xi(n,H)=\left(\Delta_1^H,\Delta_2^H, \ldots,\Delta_n^H\right),\ n\ge1.
 \end{gather*}
 
 If $H=1$, then $\ex \Delta_k^H \Delta_l^H=1$ for any $1\le k, l\le n$.
 If $H=1/2$, then $\ex \Delta_k^H \Delta_l^H=\delta_{kl}$ for any $1\le k, l\le n$.
 Therefore, for $H=1\ \det{(A_n)}=0$ and for $H=1/2\ \det{(A_n)}=1$.
It means that $\det{(A_n)}$ achieves its maximum in $H$ at the point $H=1/2$, and this  point of  maximum is unique, according to Hadamard’s inequality, and it achieves minimal value in $H$ at the point $H=1$, and this point of minimum is also unique, because it was proved in~\cite{bookBMRS} that covariance matrix of the fGn with $H\in(0,1)$ is non degenerate.
Since $\det{(A_n)}$ is a continuous function in $H$, it attains all values between 0 and 1 when $H\in(1/2,1)$.
Moreover, it was assumed in~\cite{MMRS} as a hypothesis that $\det{(A_n)}$ increases in~$H$, when $H$ increases from 0 to $1/2$, and decreases in~$H$, when $H$ increases from 1/2 to 1, but this hypothesis is still not proved.

\appendix

\section{}

\begin{lemma}\label{l:series-to-inf}
\[
\lim_{\lambda\to\infty} e^{-\lambda} \sum_{i=1}^\infty \frac{\lambda^i}{i!} \log(i+1) = \infty.
\]
\end{lemma}

\begin{proof}
Obviously, for any $N>1$
\[
\sum_{i=1}^\infty \frac{\lambda^i}{i!} \log(i+1)
\ge
\log(N+1) \sum_{i=N}^\infty \frac{\lambda^i}{i!} .
\]
Moreover,
\[
\lim_{\lambda\to\infty} e^{-\lambda} \sum_{i=0}^{N-1} \frac{\lambda^i}{i!}   = 0.
\]
Therefore,
\begin{align*}
\lim_{\lambda\to\infty} e^{-\lambda} \sum_{i=1}^\infty \frac{\lambda^i}{i!} \log(i+1)
&\ge
 \log(N+1) \lim_{\lambda\to\infty} e^{-\lambda} \sum_{i=N}^\infty \frac{\lambda^i}{i!} 
 \\
 &=  \log(N+1) \lim_{\lambda\to\infty} e^{-\lambda} \sum_{i=0}^\infty \frac{\lambda^i}{i!}= \log(N+1).
\end{align*}
Since $N>1$ is arbitrary, we get the proof.
\end{proof}

\end{document}